\documentclass[11pt]{article}
\usepackage{algorithm}
\usepackage{algorithmic}

\usepackage{subfigure}
\usepackage{epsfig,amsthm,amsmath,color, amsfonts}
\usepackage{epsfig,color}

\usepackage{framed}
\newtheorem{theorem}{Theorem}[section]
\newtheorem{corollary}[theorem]{Corollary}
\newtheorem{proposition}[theorem]{Proposition}
\newtheorem{lemma}[theorem]{Lemma}
\newtheorem{claim}[theorem]{Claim}

\theoremstyle{definition}
\theoremstyle{definition}\newtheorem{definition}[theorem]{Definition}
\theoremstyle{observation}

\newcommand{\comment}[1]{}
\newcommand{\QED}{\mbox{}\hfill \rule{3pt}{8pt}\vspace{10pt}\par}

\def\polylog{\operatorname{polylog}}
\newcommand{\ignore}[1]{}
\newcommand{\eat}[1]{}

\usepackage{times}
\usepackage[small,compact]{titlesec}
\usepackage[small,it]{caption}

\newcommand{\squishlist}{
 \begin{list}{$\bullet$}
  { \setlength{\itemsep}{0pt}
     \setlength{\parsep}{3pt}
     \setlength{\topsep}{3pt}
     \setlength{\partopsep}{0pt}
     \setlength{\leftmargin}{1.5em}
     \setlength{\labelwidth}{1em}
     \setlength{\labelsep}{0.5em} } }
\newcommand{\squishend}{
  \end{list}  }

\def\e{{\rm E}}

\def\bone{{\bf 1}}

\def\anisur#1{\marginpar{$\leftarrow$\fbox{D}}\footnote{$\Rightarrow$~{\sf #1 --Anisur}}}
\def\anisur#1{}
\def\gopal#1{}
\def\atish#1{}

\begin{document}

\title{Fast Distributed Computation in Dynamic Networks \\ via Random Walks}

\begin{titlepage}
\author{Atish {Das Sarma} \thanks{eBay Research Labs, eBay Inc., CA, USA.
\hbox{E-mail}:~{\tt atish.dassarma@gmail.com}} \and  Anisur Rahaman Molla \thanks{Division of Mathematical
Sciences, Nanyang Technological University, Singapore 637371. \hbox{E-mail}:~{\tt anisurpm@gmail.com}. Supported in part by Nanyang Technological University grant M58110000.} \and Gopal Pandurangan \thanks{Division of Mathematical
Sciences, Nanyang Technological University, Singapore 637371 and Department of Computer Science, Brown University, Providence, RI 02912, USA. \hbox{E-mail}:~{\tt gopalpandurangan@gmail.com}. Supported in part by the following research grants: Nanyang Technological University grant M58110000, Singapore Ministry of Education (MOE) Academic Research Fund (AcRF) Tier 2 grant MOE2010-T2-2-082, and a grant from the US-Israel Binational Science Foundation (BSF).}}

\date{}

\maketitle \thispagestyle{empty}


\maketitle
\begin{abstract}
The paper investigates efficient distributed computation in {\em dynamic} networks in which the
network topology changes (arbitrarily) from round to round.   
Random walks are a fundamental primitive  in a wide variety of network applications; the local
 and lightweight nature of random walks is especially useful for  providing uniform and
efficient solutions to distributed control of dynamic networks.  
Given their  applicability in dynamic networks, we  focus on developing fast distributed algorithms for performing random walks
 in such networks.  

Our first contribution is a rigorous framework for design and analysis of distributed random walk algorithms in dynamic networks. We then develop a fast distributed random walk
based algorithm that runs in $\tilde{O}(\sqrt{\tau \Phi})$ rounds\footnote{$\tilde{O}$ hides
$\polylog{n}$ factors where $n$ is the number of nodes in the
network.} (with high probability), where $\tau$ is the {\em dynamic mixing time}
and $\Phi$ is the {\em dynamic diameter} of the network respectively, and  returns a sample close to a suitably defined stationary distribution of the dynamic network.  We also apply our fast random walk algorithm to devise fast distributed algorithms for two key problems, namely,  information dissemination and decentralized computation
of spectral properties in a dynamic network.

Our next contribution is  a fast distributed algorithm 
for the fundamental problem of  information dissemination 
(also called as {\em gossip}) in a dynamic network. In gossip, or more generally,
$k$-gossip, there are $k$ pieces of information (or tokens) that are
initially present in some nodes and the problem is to disseminate the
$k$ tokens to all nodes. 
We present a random-walk based algorithm that runs in  $\tilde{O}(\min\{n^{1/3}k^{2/3}(\tau \Phi)^{1/3}, nk\})$ rounds (with high probability).  To the best of our knowledge, this is the first $o(nk)$-time  fully-distributed {\em token forwarding} algorithm that improves over  the previous-best $O(nk)$ round distributed algorithm [Kuhn et al., STOC 2010], although in an oblivious adversary model.

Our final contribution is a simple and fast distributed
algorithm for estimating the dynamic mixing time and  related spectral properties of the
underlying dynamic network.

\end{abstract}




\end{titlepage}

\section{Introduction}
\label{sec:intro}
Random walks play a central role in computer science  spanning a
wide range of areas in both theory and practice. 
Random walks are used
as an integral subroutine in a wide variety of network applications
ranging from token management and load balancing
to search, routing, information propagation and gathering,
network topology construction and building random spanning
trees (e.g., see \cite{DNP09-podc} and the references therein).
They are particularly useful in providing uniform and
efficient solutions to distributed control of dynamic networks
\cite{BBSB04, ZS06}.  Random walks are local and lightweight and
require little index or state maintenance which make them especially
attractive to self-organizing dynamic networks such as peer-to-peer,
overlay, and ad hoc wireless networks. In fact, in highly dynamic networks,
where the topology can change arbitrarily from round to round (as assumed
in this paper), extensive distributed algorithmic techniques that have been developed 
for the last few decades for {\em static} networks (see e.g., \cite{peleg,lynch,tel})
are not readily applicable. On the other hand, we would like  distributed algorithms to  work correctly and terminate even in networks that keep changing continuously over time (not assuming any eventual stabilization).
  Random walks being so simple and  very local (each subsequent step in the walk depends only on the neighbors of the current node and does not depend on the topological changes taking place elsewhere in the network) can serve as a powerful tool to design distributed algorithms for such highly dynamic networks.
However, it remains a challenge to show that one can indeed use random
walks to solve non-trivial distributed computation problems efficiently in such networks, with provable guarantees. Our paper is a step in this direction. \\
\indent A key purpose of random walks in  many of the network applications 
is to perform  node sampling.  While the sampling requirements in
different applications vary, whenever a true  sample is required from
a random walk of certain steps, typically all applications perform
the walk naively
--- by simply passing a token from one node to its neighbor: thus to
perform a random walk of length $\ell$ takes time linear in $\ell$. 
In prior work \cite{DNP09-podc, DasSarmaNPT10}, the problem of performing random walks
in time that is significantly faster, i.e., sublinear in $\ell$, was studied.
In \cite{DasSarmaNPT10}, a fast distributed random walk algorithm
was presented that ran in time sublinear in $\ell$, i.e., in  $\tilde{O}(\sqrt{\ell D})$ rounds (where $D$ is the network diameter). This algorithm used only small sized messages (i.e., it assumed the  standard CONGEST model of distributed computing \cite{peleg}). 
However, a main drawback of this result is that it applied only to {\em static}
networks. A major problem left open in \cite{DasSarmaNPT10}  is whether a similar approach
can be used to speed up random walks in dynamic networks. 

The goals of this
paper are two fold: (1) giving fast distributed algorithms for performing  random walk sampling efficiently in dynamic networks,  and (2) 
applying random walks as a key subroutine to solve non-trivial distributed computation problems in dynamic networks.
Towards the first goal, we first present a rigorous framework for studying random walks in a dynamic network (cf. Section \ref{sec:model}).  (This is necessary, since it is not immediately obvious what the output of  random walk sampling in a changing network means.)  The main purpose of  our random walk algorithm is to output a random  sample close to the ``stationary distribution" (defined precisely in Section \ref{sec:model}) of the underlying dynamic network.   Our random walk algorithms work under an oblivious adversary that fully controls the dynamic network topology, but does not know the random choices made by the algorithms (cf. Section \ref{sec:results} for a precise
problem statements and results).
   We present a fast distributed random walk
algorithm that runs in $\tilde{O}(\sqrt{\tau \Phi})$ with high probability (w.h.p.) \footnote{With high probability means with probability at least $1 - 1/n^{\Omega(1)}$, where $n$ is the number of nodes in the network.}, where $\tau$ is (an upper bound on) the dynamic mixing time
and $\Phi$ is the dynamic diameter of the network respectively (cf. Section \ref{sec:algo}).  Our algorithm uses small-sized messages only and  returns a node sample that is  ``close" to the stationary distribution of the dynamic network (assuming the stationary distribution remains fixed even as the network changes). (The precise definitions of these terms are deferred to Section \ref{sec:model}). We further extend our algorithm to efficiently perform and return $k$ independent random walk samples in  $\tilde{O}(\min\{\sqrt{k\tau \Phi}, k+\tau\})$ rounds (cf. Section \ref{sec:k-algo}). This is directly useful
in the  applications considered in this paper. 

Towards the second goal, we present two main applications of our fast random walk sampling algorithm (cf. Section \ref{sec:apps}). The first key application is a  fast  distributed algorithm 
for the fundamental problem of  {\em information dissemination}  
(also called as {\em gossip}) in a dynamic network.  In gossip, or more generally,
$k$-gossip, there are $k$ pieces of information (or tokens) that are
initially present in some nodes and the problem is to disseminate the
$k$ tokens to all nodes. In an $n$-node network, solving $n$-gossip  allows nodes to  distributively compute any computable function of their initial inputs using messages of size $O(\log n + d)$, where $d$ is the size of the input to the single node \cite{Kuhn-stoc}.  We present a random-walk based algorithm that runs in  $\tilde{O}(\min\{n^{1/3}k^{2/3}(\tau \Phi)^{1/3}, nk\})$ rounds with high probability (cf. Section \ref{sec:info-dissem}).  To the best of our knowledge, this is the first $o(nk)$-time  fully-distributed {\em token forwarding} algorithm that  improves over  the previous-best $O(nk)$ round distributed algorithm \cite{Kuhn-stoc}, albeit under an oblivious adversarial model.  A lower bound of $\Omega(nk/\log n)$  under the adaptive adversarial model of \cite{Kuhn-stoc}, was recently shown in \cite{DPRS-arxiv}; hence one cannot do substantially better than the $O(nk)$ algorithm in general under an adaptive adversary.

Our  second application is a decentralized algorithm for computing
  global metrics of the underlying dynamic network ---
 dynamic mixing time and related spectral properties (cf. Section \ref{sec:mixest}).   Such algorithms can be useful building
 blocks in the design of {\em topologically (self-)aware} dynamic networks, i.e., networks that can  monitor and regulate themselves in a decentralized fashion. For example,  efficiently computing the mixing time or the spectral gap allows  the network to monitor connectivity and expansion properties through time.

\section{Network Model and Definitions}\label{sec:model}
\subsection{Dynamic Networks}
We study a general model to describe a dynamic network with a {\em fixed} set of nodes. We consider an oblivious adversary which can make {\em arbitrary} changes to the graph topology in every round as long as the graph is {\em connected}. Such a dynamic graph process (or dynamic graph, for short) is also known as an {\em Evolving Graph} \cite{AKL08}. Suppose $V = \{v_1, v_2, \ldots,v_n\}$ be the set of nodes (vertices) and $\mathcal{G} = G_1, G_2, \ldots$ be an infinite sequence of undirected (connected) graphs on $V$. We write $G_t = (V, E_t)$ where $E_t \in 2^{V\times V} $ is the dynamic edge set corresponding  to round $t \in \mathbb{N}$. The adversary has complete control on the topology of the graph at each round, however it does not know the random choices made by the  algorithm.
In particular, in the context of random walks,  we assume that it does not know the position of the random walk in any round (however, the adversary may know the starting position).\footnote{Indeed, an adaptive adversary that always knows the current position of the random walk can easily choose graphs in each step, so that the walk never really progresses to all nodes in the network.}    Equivalently, we can assume that the adversary chooses the entire sequence $\langle G_t \rangle$ of the graph process $\mathcal{G}$ in advance before execution of the algorithm. This adversarial model has also been used  in \cite{AKL08} in their study of random walks in dynamic networks. 

We say that the dynamic graph process $\mathcal{G}$ has some property when each $G_t$ has that property. 
For technical reasons, we will assume that each graph $G_t$ is  $d$-regular and  non-bipartite. 
Later we will show that our results can be generalized to apply to non-regular graphs as well (albeit at the cost
of a slower running time). The assumption on non-bipartiteness ensures that the mixing time is well defined, however this restriction can be removed using a standard technique: adding self-loops on each vertices (e.g., see \cite{AKL08}). 
Henceforth, we assume that the dynamic graph is a {\em $d$-regular evolving graph}  unless otherwise stated (these two terms will be used interchangeably).  Also we will assume that each $G_t$ is non-bipartite (and connected).

\subsection{Distributed Computing Model}
\label{sec:distmodel}
We model the communication network as an $n$-node dynamic graph process $\mathcal{G} = G_1, G_2, \ldots$. Every  node has limited initial knowledge. Specifically, assume that each node is associated with a distinct identity number (id). (The node ids are of size $O(\log n)$.) At the beginning of the computation, each node $v$ accepts as input its own identity number and the identity numbers of its neighbors in $G_1$. The node may also accept some additional inputs as specified by the problem at hand (in particular, we assume that all nodes know $n$). The nodes are allowed to communicate through the edges of the graph $G_t$  in each round $t$. We assume that the communication occurs in  synchronous  {\em rounds}. 
In particular, all the nodes wake up simultaneously at the beginning of round $1$, and from this point on the nodes always know the number of the current round. 
We will use only small-sized messages. In particular, at the beginning of each round $t$, each node $v$ is allowed to send a message of size $B$ bits (typically $B$ is assumed to be $O(\polylog n)$) through each edge $e = (v, u) \in E_t$ that is adjacent to $v$.  The message  will arrive to $u$ at the end of the current round. 
This is a standard model of distributed computation known as the {\em CONGEST(B) model} \cite{peleg, PK09} and has been attracting a lot of research attention during last two decades (e.g., see \cite{peleg} and the references therein).
For the sake of simplifying our analysis, we assume that $B = O(\log^2 n)$, although this is generalizable.\footnote{It turns out that the per-round congestion in any edge  in our random walk algorithm is $O(\log^2 n)$ bits  w.h.p.  Hence assuming this bound for $B$ ensures that the random walks can never be delayed due to congestion. This simplifies the correctness proof of our random walk algorithm (cf. Section \ref{sec:correctness}.)} 

There are several measures of efficiency of distributed algorithms,
but we will focus on one of them, specifically, {\em the
running time}, i.e. the number of {\em rounds} of distributed
communication. (Note that the computation that is performed by the
nodes locally is ``free'', i.e., it does not affect the number of
rounds.)

\subsection{Random Walks in a Dynamic Graph}
\label{sec:rwd}
Throughout, we assume the {\em simple   random walk} in an undirected graph: In each step, the walk goes from the current node to a random neighbor, i.e., from the current node $v$, the probability to move in the next step to a neighbor $u$ is $\Pr(v,u) = 1/d(v)$  for $(v,u) \in E$ and $0$ otherwise  ($d(v)$ is the degree of $v$).

A {\em simple random walk} on dynamic graph $\mathcal{G}$ is defined as follows: assume that at time $t$ the walker is at node $v \in V$, and let $N(v)$ be the set of neighbors of $v$ in $G_t$, then the walker goes to one of its neighbors from $N(v)$ uniformly at random.  


Suppose we have a random walk $v_0 \rightarrow v_1 \rightarrow \ldots \rightarrow v_t$ on a dynamic graph $\mathcal{G}$, where $v_0$ is the starting vertex. Then we get a probability distribution $P_t$ on $v_t$ starting from the initial distribution $P_0$ on $v_0$. We say that the distribution $P_r$ is stationary (or steady-state) for the graph process $\mathcal{G}$ if $P_{t+1} = P_t$ for all $t \ge r$. It is known that for every (undirected) static graph $G$, the distribution $\pi(v) = d(v)/2m$ is stationary. In particular, for a regular graph the stationary distribution is the uniform distribution.
The {\em mixing time} of a random walk on a static graph $G$ is the time $t$ taken to reach ``close'' to the stationary distribution of the graph (see Definition~\ref{def:mix-static} in Section \ref{mixing_time} below). Similar to the static case, for a $d$-regular evolving graph, it is easy to verify that the stationary distribution is the uniform distribution.  Also, for a $d$-regular evolving graph,  the notion of {\em dynamic mixing time}   (formally defined in Section \ref{mixing_time}) is similar to the static case  and is well defined due to the monotonicity property of distributions (cf. Lemma~\ref{lem:monotonicity} in Section \ref{mixing_time}). 
We show (cf. Theorem~\ref{thm:mixtime} in the next Section) that the dynamic mixing time is  bounded by $O(\frac{1}{1-\lambda}\log n)$ rounds, where $\lambda$ is an upper bound of the second largest eigenvalue in absolute value of any graph in $\mathcal{G}$.  Note that $O(\frac{1}{1-\lambda}\log n)$ is also an upper bound on the mixing time of the graph having $\lambda$ 
as its second largest eigenvalue and hence the dynamic mixing time is upper bounded by the worst-case mixing time of any graph in $\mathcal{G}$, which will be (henceforth) denoted by  $\tau$.  
Since the second eigenvalue of the transition matrix of any regular graph is bounded by $1-1/n^2$ (cf. Corollary~\ref{cor:second-eigen-bound}), this implies that $\tau$ of a $d$-regular evolving graph is bounded by $\tilde{O}(n^2)$ (cf. Section \ref{mixing_time}). 
In general, the dynamic mixing time can be significantly smaller than this bound, e.g., when all graphs in $\mathcal{G}$ have $\lambda$ bounded from above by a constant (i.e.,  they are expanders ---  such dynamic graphs occur in applications e.g., \cite{JGPE:soda12,Kuhn-stoc}), the dynamic mixing time is $O(\log n)$. 

Another parameter affecting the efficiency of distributed computation in a dynamic graph is its dynamic diameter (also called flooding time, e.g., see~\cite{BCF09, CMMPS08}).  
The {\em dynamic diameter} (denoted by $\Phi$) of an $n$-node dynamic graph $\mathcal{G}$ is the worst-case time (number of rounds) required to broadcast a piece
of information from any given node to all $n$-nodes. 
The dynamic diameter  can be much larger than the diameter ($D$) of any (individual) graph $G_t$.

\section{Mixing Time of Regular Dynamic Graph} \label{mixing_time}
We have discussed the notion of random walk, probability distribution of a random walk , mixing time etc. on (dynamic) graph in Section~\ref{sec:rwd} above. Here we formally define those notions.  
\begin{definition} [Distribution vector]
Let $\pi_x(t)$ define the probability distribution vector reached after $t$ steps when the initial distribution starts with probability $1$ at node $x$. Let $\pi$ denote the stationary distribution vector.
\end{definition}

\begin{definition}\label{def:mix-static} [$\tau^x(\epsilon)$ ($\epsilon$-near mixing time for source $x$), $\tau^x_{mix}$ (mixing time for source $x$) and $\tau_{mix}$ (mixing time)]
Define $\tau^x(\epsilon) = \min t : ||\pi_x(t) - \pi|| < \epsilon$. Define $\tau^x_{mix} = \tau^x(1/2e)$. Define $\tau_{mix} = \max_{x} \tau^x_{mix}$. ($\epsilon$ is a small constant). 
\end{definition}

We define the {\em dynamic mixing time} of a $d$-regular evolving graph $\mathcal{G} = G_1, G_2, \ldots$ as the maximum time taken for a simple random walk starting from any node to reach {\em close to} the uniform distribution on the vertex set. Therefore the definition of dynamic mixing time  is similar to the static case.   Let $\tau$ be the maximum mixing time of any (individual) graph $G_t$ in $\mathcal{G}$. (Note that  $\tau$ is bounded by $\tilde{O}(n^2)$ --- follows from Theorem \ref{thm:mixtime} and Corollary \ref{cor:second-eigen-bound}).  We show that dynamic mixing time is well defined due to Theorem \ref{thm:mixtime} and monotonicity property of distribution (cf. Section \ref{sec:monotonicity}).   

\begin{definition}\label{def:mix-dynamic}[Dynamic mixing time]
Define $\tau^x(\epsilon)$ ($\epsilon$-near mixing time for source $x$) is $\tau^x(\epsilon) = \min t : ||\pi_x(t) - \pi|| < \epsilon$. Note that $\pi_x(t)$ is the probability distribution on the graph $G_t$ in the dynamic graph process $\{ G_t : t\geq 1\}$ when the initial distribution ($\pi_x(1)$) starts with probability $1$ at node $x$ on $G_1$. Define $\tau^x_{mix}$ (mixing time for source $x$) $ = \tau^x(1/2e)$ and $\tau_{mix} = \max_{x} \tau^x_{mix}$. The dynamic mixing time is upper bounded by   $\tau = \max \{$mixing time of all the static graph $G_t : t \geq 1\}$. Notice that $\tau \geq \tau_{mix}$ in general. 
\end{definition} 

It is known that simple random walks on regular, connected, non-bipartite static graph have mixing time $O(\frac{\log n}{1 - \lambda_2})$, where $\lambda_2$ is the second largest eigenvalue in absolute value of the graph. Interestingly, it turns out that similar result holds for $d$-regular, connected, non-bipartite evolving graphs. We show that the mixing time of a simple random walk on a dynamic graph $\mathcal{G} = G_1, G_2, \ldots$ is $O(\frac{\log n}{1 - \lambda})$, where $\lambda$ is an upper bound of the second largest eigenvalue in absolute value of the graphs $\{G_t : t \geq 1\}$. 

\begin{lemma}\label{lem:lemma2}
Let $G$ be an undirected connected non-bipartite $d$-regular graph on $n$ vertices and $p=(p_1, \ldots ,p_n)$ be any probability distribution on its vertices. Let $A_G$ be the transition matrix of a simple random walk on $G$. Then,
$$ \bigl\Vert pA_G - \frac{\textbf{1}}{n} \bigr\Vert \leq \bar{\lambda_2}\cdot \bigl\Vert p - \frac{\textbf{1}}{n} \bigr\Vert $$ where $\bar{\lambda_2} = \max_{i=2,\ldots,n} \lvert \lambda_i \rvert = \max \{\lambda_2, - \lambda_n\}$ be the second largest eigenvalue in absolute value.
\end{lemma}
\begin{proof}
Let $X_1, X_2, \ldots ,X_n$ be an orthonormal set of eigenvectors of $A_G$ with corresponding eigenvalues $\lambda_1 \geq \lambda_2 \geq \ldots \geq \lambda_n$. Since $A_G$ is symmetric stochastic matrix, $\lambda_1 = 1$ and $X_1 = \frac{\textbf{1}}{\sqrt{n}}$ and all eigenvectors and eigenvalues are real. Clearly, 
$$ \bigl\Vert pA_G - \frac{\textbf{1}}{n} \bigr\Vert = \bigl\Vert pA_G - \frac{\textbf{1}}{n}A_G \bigr\Vert = \bigl\Vert (p - \frac{\textbf{1}}{n}) A_G \bigr\Vert $$
Since $p$ is a probability distribution, we can write it as $p = \beta_1X_1 + \beta_2X_2 + \ldots + \beta_n X_n$, where $\beta_1, \beta_2, \ldots ,\beta_n \in \mathbb{R}$. Then $\beta_1 = p\cdot X_1^T = \frac{1}{\sqrt{n}} \sum_{i}p_i = \frac{1}{\sqrt{n}}$, so that $\beta_1 X_1 = (\frac{1}{n}, \frac{1}{n}, \ldots ,\frac{1}{n}) $. Therefore, $p - \frac{\textbf{1}}{n} = \sum_{i=2}^n \beta_i X_i $. Hence, $$ \bigl\Vert p - \frac{\textbf{1}}{n} \bigr\Vert = \sqrt{\sum_{i=2}^n \beta_i^2} $$
Furthermore, 
\begin{align*} \bigl\Vert (p - \frac{\textbf{1}}{n}) A_G \bigr\Vert & = \bigl\Vert \sum_{i=2}^n \beta_i X_i A_G \bigr\Vert \\
& = \bigl\Vert \sum_{i=2}^n \lambda_i \beta_i X_i \bigr\Vert \\ 
& = \sqrt{\sum_{i=2}^n \lambda_i^2 \beta_i^2} \\
& \leq \max_{i=2,\ldots,n} \vert \lambda_i \vert \cdot \sqrt{\sum_{i=2}^n \beta_i^2} \\
& = \bar{\lambda_2}\cdot \bigl\Vert p - \frac{\textbf{1}}{n} \bigr\Vert 
\end{align*}
Thus, $$\bigl\Vert pA_G - \frac{\textbf{1}}{n} \bigr\Vert \leq \bar{\lambda_2}\cdot \bigl\Vert p - \frac{\textbf{1}}{n} \bigr\Vert  $$
\end{proof}


An immediate corollary follows from the previous lemma: 
\begin{corollary}\label{cor:cor1}
Let $\mathcal{G} = G_1, G_2, \ldots $ be a sequence of undirected connected non-bipartite $d$-regular graphs on the same vertex set $V$. If $p_0$ is the initial probability distribution on $V$ and we perform a simple random walk on $\mathcal{G}$ starting from $p_0$, then the probability distribution $p_t$ of the walk after $t$ steps satisfies, 
$$\bigl\Vert p_t - \frac{\textbf{1}}{n} \bigr\Vert \leq {\lambda}^t \bigl\Vert p_0 - \frac{\textbf{1}}{n} \bigr\Vert $$ where $\lambda$ is an upper bound on the second largest eigenvalue in absolute value of the graphs $\{G_t : t \geq 1\}$. 
\end{corollary}

\begin{theorem}\label{thm:mixtime}
For any $d$-regular connected non-bipartite evolving graph $\mathcal{G}$, the dynamic mixing time of a simple random walk on $\mathcal{G}$ is bounded by $O(\frac{1}{1- {\lambda}}\log n)$, where $\lambda$ is an upper bound of the second largest eigenvalue in absolute value of any graph in $\mathcal{G}$.
\end{theorem}
\begin{proof}
Let the random walk starts from a given vertex with distribution $p_0 = (1, 0, \ldots, 0)$. From Corollary~\ref{cor:cor1} and the fact that $\bigl\Vert p_0 - \frac{\textbf{1}}{n} \bigr\Vert = O(1)$ we have,
$$\bigl\Vert p_t - \frac{\textbf{1}}{n} \bigr\Vert \leq {\lambda}^t $$
So for $t = O(\frac{1}{1 - {\lambda}} \log n)$ gives $\bigl\Vert p_t - \frac{\textbf{1}}{n} \bigr\Vert \leq \frac{1}{n^{O(1)}}$.  
\end{proof}

\begin{lemma}\label{lem:eigen-bound}
Let $G$ be an undirected connected $d$-regular graph on $n$ vertices. Let $A_G$ be the transition matrix of $G$. If $\lambda_1 \geq \lambda_2 \geq \ldots \geq \lambda_n$ are the eigenvalues of $A_G$, then $\lambda_1 = 1$ and for $i \geq 2, \lambda_i \leq 1 - \frac{1}{d D n}$, where $D$ is the diameter of the graph.  
\end{lemma}
\begin{proof}
The normalized eigenvector corresponding to $\lambda_1 = 1$ is $X_1 = \frac{1}{\sqrt{n}}(1, 1, \ldots , 1)$. Consider any normalized real eigenvector $X \perp X_1$ with its eigenvalue $\lambda$. Hence, $\sum_{i=1}^n x_{i}^2 = 1$ and $ \sum_{i=1}^n x_i = 0$, where $X = (x_1, x_2, \ldots ,x_n)$. Let $x_t$ and $x_s$ be the respectively largest and smallest co-ordinates of $X$. Clearly $ x_t \geq 1/\sqrt{n}$ and $x_s < 0$. Let $s=v_1 \to v_2 \to \ldots \to v_k =t$ be the vertices on a shortest path from $s$ to $t$ in $G$. Consider $X(\mathbb{I} - A_G)X^T$. Then,
\begin{align*}
1 - \lambda & = X(\mathbb{I} - A_G)X^T = \frac{1}{d} \cdot \sum_{\{i, j\} \in E(G)} (x_i - x_j)^2 \\
& \geq \frac{1}{d} \cdot \sum_{i=1}^{k-1} (x_{v_i} - x_{v_{i+1}})^2 \\
& \geq \frac{1}{d(k - 1)} \cdot \left( \sum_{i = 1}^{k-1} x_{v_i} - x_{v_{i+1}} \right) && [\text{by Cauchy-Schwarz inequality}] \\
& = \frac{1}{d (k-1)} \cdot (x_s - x_t)^2 \\
& \geq \frac{1}{d D n}
\end{align*}
where $k-1 \leq D$, the diameter of the graph.
\end{proof}

\begin{corollary}\label{cor:second-eigen-bound}
$(1- \frac{1}{n^2})$ is an upper bound of the second largest eigenvalue $\bar{\lambda_2}$ of the transition matrix of any undirected connected regular graph on $n$-vertices. 
\end{corollary}
\begin{proof}
This follows from the Lemma~\ref{lem:eigen-bound} and the fact that the diameter of any connected regular graph is bounded by $O(\frac{n}{d})$. 
\end{proof}

\subsection{Monotonicity property of the distribution vector}\label{sec:monotonicity}
Let $\pi_x(t)$ define the probability distribution vector of a simple random walk reached after $t$ steps when the initial distribution starts with probability $1$ at node $x$. Let $\pi$ denote the stationary distribution vector. We show in the following lemma that the vector $\pi_x(t)$ gets closer to $\pi$ as $t$ increases. 
\begin{lemma} 
\label{lem:monotonicity}
$||\pi_x(t+1) - \pi|| \leq  ||\pi_x(t) - \pi||$.
\end{lemma}
\begin{proof}
We need to show that the definition of mixing times are consistent, i.e. monotonic in $t$, the walk length of the random walk. Let $A$ be the transition matrix of a simple random walk on a $d$-regular evolving graph $\mathcal{G}$ which in fact changes from round to round. The entries $a_{ij}$ of $A$ denotes the probability of transitioning from node $i$ to node $j$. The monotonicity follows from the fact that for any transition matrix $A$ of any regular graph and for any probability distribution vector $p$, $$||(p - \frac{\mathbf{1}}{n}) A|| < ||p - \frac{\mathbf{1}}{n}||.$$ This result follows from the above Lemma~\ref{lem:lemma2} and the fact that $\bar{\lambda_2} < 1$.

Let $\pi$ be the stationary distribution of the matrix $A$. Then $\pi = (\frac{1}{n}, \frac{1}{n}, \ldots, \frac{1}{n})$. This implies that if $t$ is $\epsilon$-near mixing time, then $||p A^t - \pi|| \leq \epsilon$, by definition of $\epsilon$-near mixing time. Now consider $||p A^{t+1} - \pi||$. This is equal to $||p A^{t+1} - \pi A||$ since $\pi A = \pi$.  However, this reduces to $||(p A^{t} - \pi) A|| < ||p A^t - \pi|| \leq \epsilon$. It follows that $(t+1)$ is $\epsilon$-near mixing time and $||p A^{t+1} - \pi|| < ||p A^t - \pi||$.
\end{proof}

\section{Problem Statements and Our Results}\label{sec:results}
We formally state the problems and our main results.
\paragraph{The Single Random Walk problem.} Given a $d$-regular evolving graph $\mathcal{G} = (V, E_t)$ and a starting node $s \in V$, our goal is to devise a fast distributed  {\em random walk} algorithm such that, at the end, a destination node, sampled from a $\tau$-length walk, outputs the source node's ID (equivalenty, one can require $s$ to output the destination node's ID), where $\tau$ is (an upper bound on) the dynamic mixing time of $\mathcal{G}$ (cf. Section \ref{mixing_time}), under the assumption that $\mathcal{G}$ is modified by an oblivious adversary (cf. Section \ref{sec:model}). Note that this
distribution will be ``close" to the stationary distribution of $\mathcal{G}$ (stationary distribution and $\tau$
are both well-defined  --- cf. Section  \ref{mixing_time}).  
Since we are assuming a $d$-regular evolving graph,  our goal is to sample from (or close to) the uniform distribution (which is the stationary distribution) using as few rounds as possible. Note that we would like to sample fast via  random walk --- this is also very important for
the applications considered in this paper. On the other hand, if one had to simply get a  uniform random sample, it can be accomplished by other means, e.g.,
it is easy to obtain it in $O(\Phi)$ rounds (by using flooding).

For clarity, observe that the following naive algorithm solves the
above problem in $O(\tau)$ rounds: The walk of length $\tau$ is
performed by sending a token for $\tau$ steps, picking a random
neighbor in each step. Then, the destination node $v$ of this walk
outputs the ID of $s$.
Our goal is to perform such sampling with significantly less number
of rounds, i.e., in time that is sublinear in $\tau$, in the CONGEST model, and  using  random walks rather than naive flooding techniques. As mentioned earlier this is needed for the applications discussed in this paper. Our  result is as follows. 
\begin{theorem}\label{thm:maintheorem}
The algorithm {\sc Single-Random-walk} (cf. Algorithm \ref{alg:single-random-walk}) solves the Single Random Walk problem in a dynamic graph and with high probability finishes in $\tilde{O}(\sqrt{\tau \Phi})$ rounds. 
\end{theorem}

The above algorithm assumes that nodes have knowledge of $\tau$ (or at least some good estimate of it). (In many applications, it is easy to have a good estimate of  $\tau$  when there is knowledge of the structure of the individual graphs --- e.g.,  each $G_t$ is an expanders as in \cite{JGPE:soda12,FOCS2001}.)  Notice that in the worst case the value of $\tau$ is $\tilde{\Theta}(n^2)$, and hence this bound can be used even if nodes have no knowledge. Therefore putting $\tau = \tilde{\Theta}(n^2)$ in the above Theorem \ref{thm:maintheorem}, we see that our algorithm samples a node from the uniform distribution through a random walk in $\tilde O(n \sqrt{\Phi})$ rounds w.h.p. Our algorithm can also be generalized to work for non-regular evolving graphs also (cf. Section \ref{sec:nonregular}).

We also consider the following extension of the Single Random Walk problem,
called  \textbf{the $k$ Random Walks problem}: We have $k$ sources $s_1, s_2, ..., s_k$ (not necessarily distinct) and we want each of 
the $k$ destinations to output an ID of its corresponding source, assuming that each source initiates an independent
random walk of length $\tau$. (Equivalently, one can ask each source to output the ID of its corresponding destination.) The goal is to output all the ID's in as few rounds as possible. We show that:

\begin{theorem}\label{thm:kwalks} The algorithm {\sc Many-Random-Walks} (cf. Algorithm \ref{alg:many-random-walk}) solves the $k$ Random Walks problem in a dynamic graph and with high probability finishes in
$\tilde O\left(\min(\sqrt{k\tau \Phi}, k+\tau)\right)$ rounds.
\end{theorem}

\paragraph{Information dissemination (or $k$-gossip) problem.} In $k$-gossip, initially $k$ different tokens are assigned to a set $V$ of $n (\ge k)$ nodes. A node may have more than one token. The  goal is to disseminate all the $k$ tokens to all the $n$ nodes. We present a fast distributed randomized algorithm for $k$-gossip in a dynamic network. Our algorithm uses {\sc Many-Random-Walks} as a key subroutine; to the best of our knowledge, this is the first subquadratic time fully-distributed {\em token forwarding} algorithm. 
\begin{theorem}\label{thm:token-bound}
The algorithm {\sc K-Information-Dissemination} (cf. Algorithm~\ref{alg:token-dissemination}) solves {\em $k$-gossip} problem in a dynamic graph with high probability in $\tilde{O}(\min\{n^{1/3}k^{2/3}(\tau \Phi)^{1/3}, nk\})$ rounds. 
\end{theorem}

\paragraph{Mixing time estimation.} Given a dynamic network, we are interested in  (approximately) computing the dynamic mixing time, assuming that the mixing time of the (individual) graphs do not change. We present an efficient distributed algorithm  for estimating the mixing time. In
particular, we show the following result where $\tau^x_{mix}$ is the dynamic mixing time with respect to a starting node $x$. We formally define these notions in Section~\ref{mixing_time}.
This gives an alternative algorithm to the only previously known
approach by Kempe and McSherry \cite{kempe}  that can be used to estimate
$\tau^x_{mix}$ in a {\em static} graph in $\tilde O(\tau^x_{mix})$ rounds. 
\begin{theorem}\label{thm:complexity_bound_mixing_time}
Given connected $d$-regular evolving graphs with dynamic diameter $\Phi$, a node $x$ can find, in $\tilde{O}(n^{1/4}\sqrt{\Phi \tau^x(\epsilon)})$ rounds, a time
$\tilde{\tau}^x_{mix}$ such that $\tau^x_{mix}\leq \tilde{\tau}^x_{mix}\leq \tau^x(\epsilon)$, where $\epsilon = \frac{1}{6912e\sqrt{n}\log n}$.
\end{theorem}

\section{Related Work and Technical Overview}\label{sec:related}
\paragraph{Dynamic networks.} As a step towards understanding the fundamental computational power in
dynamic networks, recent studies (see e.g., \cite{dynamic-survey, Kuhn-stoc, kuhn-podc, DPRS-arxiv} and the references therein) have investigated dynamic networks in which the network
topology changes arbitrarily from round to round. In the
worst-case model that was studied by Kuhn, Lynch, and
Oshman~\cite{Kuhn-stoc}, the communication links for
each round are chosen by an online adversary, and nodes  do not know
who their neighbors for the current round are before they broadcast
their messages. 
Unlike prior
models on dynamic networks, the model of~\cite{Kuhn-stoc} (like ours) does
not assume that the network eventually stops changing; therefore it requires
that the {\em algorithms work correctly and terminate even in networks that
change continually over time}.

The work of~\cite{AKL08} studied the {\em cover time} of  random walks in an evolving graph (cf. Section \ref{sec:model}) in an oblivious adversarial model. 
In a {\em regular} evolving graph, they show that the cover time is always polynomial, while this is not true in
general if the graph is not regular --- the cover time can be exponential.  
However, they show that a lazy random walk (i.e., walk with self loops)  has polynomial cover time on all graphs. 
We also use a similar strategy to show that our distributed random walk algorithms can work on non-regular graphs also, albeit at the cost of an increase in run time.
While the work of \cite{AKL08} addressed the cover time of random walks
on dynamic graphs, this paper is concerned with distributed algorithms for computing random walk samples fast with the goal towards applying it to fast distributed computation problems in dynamic networks. 

Recently, the work of \cite{clementi-podc12}, studies the flooding time of {\em Markovian} evolving dynamic graphs, a special class of evolving graphs.

\paragraph{Distributed random walks.} Our fast distributed random walk algorithms are based on previous such algorithms designed for
{\em static} networks \cite{DNP09-podc, DasSarmaNPT10}. These were the first sublinear  (in the length of the walk) time algorithms for performing random walks in graphs.  The algorithm of \cite{DasSarmaNPT10} performed a random walk of length $\ell$  in $\tilde{O}(\sqrt{\ell D})$  rounds (with high probability) on an undirected  network, where $D$ is the diameter of the network. 
(Subsequently, the algorithm of \cite{DasSarmaNPT10} was shown to be almost time-optimal (up to polylogarithmic factors) in \cite{NanongkaiDP11}.)
The general high-level idea   of the above algorithm is using  a few short walks in the
beginning (executed in parallel) and then carefully concatenating these
walks together later as necessary.  
A main contribution of the present work is showing that building on the approach
of \cite{DasSarmaNPT10} yields speed up in random walk computations even in dynamic networks. However, there are some challenging
technical issues to overcome in this extension given the continuous dynamic nature (cf. Section \ref{sec:algo}). 
One key  technical lemma (called the {\em Random walk visits Lemma}) that was used
to show the almost-optimal run time of $\tilde{O}(\sqrt{\ell D})$ does not directly apply to dynamic networks. 
In the static setting, this lemma
gives a bound on the number of times  any node is visited in an $\ell$-length walk, for any  length that is not much larger than the cover time.  More precisely, the lemma states that w.h.p. any node $x$ is visited at most $\tilde{O}(d(x)\sqrt{\ell})$ times, in an $\ell$-length walk from any starting node ($d(x)$ is the degree of $x$).
In this paper, we show that a similar bound applies to an $\ell$-length random walk
on any $d$-regular evolving graph (cf. Lemma \ref{lem:visit-bound}).
A key ingredient in the above proof is  showing that a technical result due to Lyons \cite{Lyons}
can be made to work on an evolving graph.

Other recent work involving multiple random walks in {\em static} networks, but in 
different settings include Alon et. al.~\cite{AAKKLT}, Els{\"a}sser
et. al.~\cite{berenbrink+ceg:gossip},  and Cooper et al. \cite{frieze}.
\paragraph{Information spreading.} The main application of our random walks algorithm is an improved algorithm for 
information spreading or gossip in dynamic networks. To the best of our knowledge, it gives the first
subquadratic, fully distributed, token forwarding algorithm in dynamic networks, partially  answering an open question raised in \cite{DPRS-arxiv}.  Information spreading is a fundamental
primitive in networks which has been extensively studied  (see e.g., \cite{DPRS-arxiv} and the references therein). Information spreading can be used to solve other problems such as
broadcasting and leader election. 

This  paper's focus is on {\em token-forwarding} algorithms, which do not manipulate tokens in any way other than
storing and forwarding them.  Token-forwarding algorithms are simple,
often easy to implement, and typically incur low overhead.  \cite{Kuhn-stoc} showed that under their adversarial
model, $k$-gossip can be solved by token-forwarding in $O(nk)$ rounds,
but that any deterministic online token-forwarding algorithm needs
$\Omega(n \log k)$ rounds. In \cite{DPRS-arxiv}, an almost matching lower bound of 
$\Omega(nk/\log n)$ is shown.
The above lower bound indicates that one cannot obtain efficient (i.e.,
subquadratic) token-forwarding algorithms for gossip in the
adversarial model of~\cite{Kuhn-stoc}.   This motivates considering
other weaker (and perhaps more realistic) models of dynamic networks.

\cite{DPRS-arxiv} presented a polynomial-time offline {\em centralized} token-forwarding algorithm that
solves the $k$-gossip problem on an $n$-node dynamic network in
$O(\min\{nk, n \sqrt{k \log n}\})$ rounds with high probability. This is the first known {\em subquadratic} time
token-forwarding algorithm but it is not distributed, and furthermore, the centralized algorithm needs
to know the complete evolution of the dynamic graph in advance. 
It was left open in \cite{DPRS-arxiv}
whether one can obtain a fully-distributed  and localized algorithm that also does not know anything
about how the network evolves.
In this paper, we resolve this open question in the affirmative. Our  algorithm  runs in  $\tilde{O}(\min\{n^{1/3}k^{2/3}(\tau \Phi)^{1/3}, nk\})$ rounds with high probability.  This is significantly faster than the $O(nk)$-round
algorithm of \cite{Kuhn-stoc} as well as the above centralized algorithm of \cite{DPRS-arxiv} when $\tau$ and $\Phi$ are not too large. Note that $\Phi$ is bounded by $O(n)$  and in regular graphs $\tau$ is  $O(n^2)$ ($O(n^3)$ in general graphs) and so in general, our bounds cannot be better than $O(nk)$. 

We note that  an alternative approach
based on network coding was due to
~\cite{haeupler:gossip,haeupler+k:dynamic}, which achieves an
$O(nk/\log n)$ rounds using $O(\log n)$-bit messages (which is not
significantly better than the $O(nk)$ bound using token-forwarding),
and $O(n + k)$ rounds with large message sizes (e.g., $\Theta(n\log n)$ bits).  
It thus follows that for large token and message sizes
there is a factor $\Omega(\min\{n,k\}/\log n)$ gap between
token-forwarding and network coding. We note that in our model we
allow only one token per edge per round and thus our bounds hold
regardless of the token size.

\section{Algorithm for Single Random Walk}\label{sec:algo}
\subsection{Description of the Algorithm}
We develop an algorithm called {\sc Single-Random-Walk} (cf. Algorithm~\ref{alg:single-random-walk}) for $d$-regular evolving graph ($\mathcal{G} = (V, E_t)$).  The algorithm performs a random walk of length $\tau$ (the dynamic mixing time of $\mathcal{G}$ --- cf. Section \ref{sec:rwd}) in order to sample a destination from  (close to) the uniform distribution on the vertex set $V$. 

The high-level idea of the algorithm is to perform ``many"  short random walks in parallel and later ``stitch" the short walks to get the desired walk of length $\tau$. In particular, we perform the algorithm in two phases, as follows. For simplicity we call the messages used in Phase~1 as ``coupons" and in Phase~2 as ``tokens". 
In Phase 1, we perform $d$ (degree of the graph) ``short"  (independent) random walks of length $\lambda$ (to bound the running time correctly, we show later that we do short walks of length approximately $\lambda$, instead of $\lambda$) from each node $v$, where $\lambda$ is a parameter whose value will be fixed in the analysis. This is done simply  by forwarding $d$ ``coupons" having the ID of $v$ from $v$ (for each node $v$) for $\lambda$ steps via random walks. 

In Phase 2, starting at source $s$, we ``stitch" (see Figure \ref{fig:connector}) some of short walks prepared in Phase 1 together to form a longer walk. The algorithm starts from $s$ and randomly picks one coupon distributed from $s$ in Phase 1. We now discuss how to sample one such coupon randomly and go to the destination vertex of that coupon. This can be done easily as follows: In the beginning of Phase~1, each node $v$ assigns a coupon number for each of its $d$ coupons. At the end of Phase 1, the coupons originating at $s$ (containing ID of $s$ plus a coupon number) are distributed throughout the network (after Phase 1). When  a coupon needs to be  sampled, node $s$  chooses a random coupon number (from the unused set of coupons) and informs the destination node (which will be the next stitching point) holding the coupon $C$  through flooding. 
Let $C$ be the  sampled coupon and $v$ be the destination node of $C$. $s$ then sends a ``token" to $v$ (through flooding) and $s$ deletes coupon $C$ (so that $C$ will not be sampled again next time at $s$, otherwise, randomness will be destroyed). The process then repeats. That is, the node $v$ currently holding the token samples one of the coupons it distributed in Phase 1 and forwards the token to the destination of the sampled coupon, say $v'$. Nodes $v, v'$ are called ``connectors" - they are the endpoints of the short walks that are stitched. A crucial observation is that the walk of length $\lambda$ used to distribute the corresponding coupons from $s$ to $v$ and from
$v$ to $v'$ are independent random walks. Therefore, we can stitch them to get a random walk of length $2\lambda$. We therefore can generate a random walk of length $3\lambda, 4\lambda, \ldots $ by repeating this process. We do this until we have completed more than $\tau - \lambda$ steps. Then, we complete the rest of the
walk by doing the naive random walk algorithm. 

\begin{figure}[h]
\centering
\includegraphics[width=0.98\linewidth]{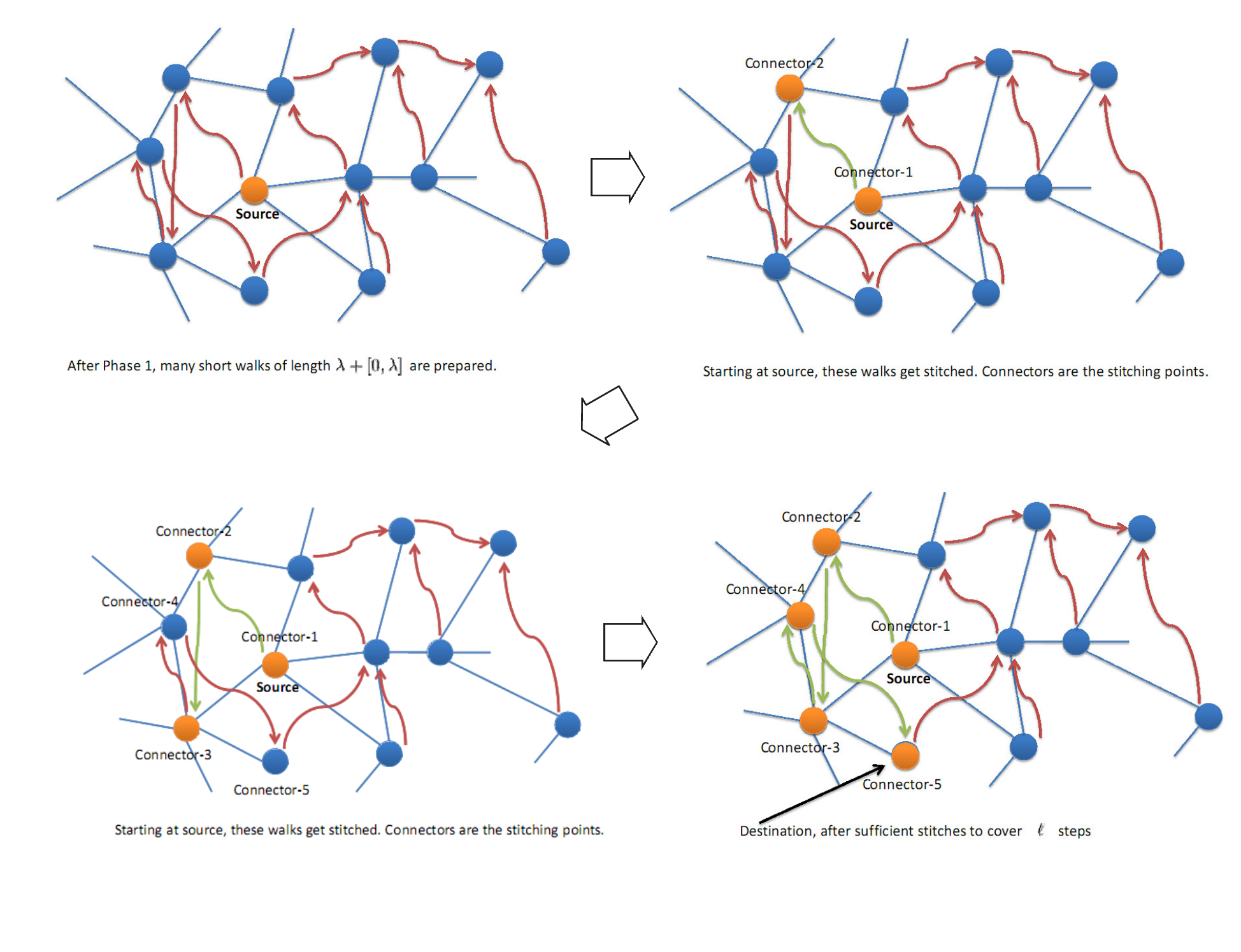}
\caption{Figure illustrating the Algorithm of stitching short walks
together.} \label{fig:connector}
\end{figure}

To understand the intuition behind this algorithm, let us analyze its running time. First, we claim that Phase 1 needs $O(\lambda)$(see Lemma~\ref{lem:phase1}) rounds with high probability. Recall that, in Phase 1, each node prepares $d$ independent random walks of length $\lambda$ (approximately). We start with $d =\deg(v)$  coupons from each node $v$ at the same time, each edge in the current graph should receive two coupons in the average case. In other words, there is essentially no congestion (i.e., not too many coupons are sent through the same edge). Therefore sending out (just) $d$ coupons from each node for $\lambda$ steps will take $O(\lambda)$ rounds in expectation. This argument can be modified to show that we need $O(\lambda)$ rounds with high probability in our model (see full proof of the Lemma~\ref{lem:phase1}). Now by the definition of dynamic diameter, flooding takes $\Phi$ rounds. We show that sample a coupon can be done in $O(\Phi)$ rounds (cf. Lemma~\ref{lem:lemma2.3}) and it follows that Phase 2 needs $\tilde O(\Phi \cdot \tau/\lambda)$ rounds. Therefore, the algorithm needs $\tilde{O}(\lambda+ \Phi \cdot \tau/\lambda)$ which is $\tilde{O}(\sqrt{\tau \Phi})$ when we set $\lambda =\sqrt{\tau \Phi}$. 

The reason the above algorithm for Phase 2 is incomplete is that it is possible that $d$ coupons are not enough: We might forward the token to some node $v$ many times in
Phase 2 and all coupons distributed by $v$ in the first phase are deleted. (In other words, $v$ is chosen as a connector node many times, and all its coupons have been exhausted.)
If this happens then the stitching process cannot progress. To fix this problem, we will show (in the next section) an important property of the random walk which
says that a random walk of length $O(\tau)$ will visit each node $v$ at most $\tilde{O}(\sqrt{\tau} d)$ times (cf. Lemma \ref{lem:visit-bound}). But this bound is not enough to get the desired running time, as it does not say anything about the distribution of the connector nodes. We use the following idea to overcome it: Instead of nodes performing walks of length $\lambda$, each such walk $i$ do a walk of length $\lambda + r_i$ where $r_i$ is a random number in the range $[0, \lambda-1]$. Since the random numbers are independent for each walk, each short walks are now of a random length in the range $[\lambda, 2\lambda-1]$. This modification is needed to claim that each node will be visited as a connector only $\tilde{O}(\sqrt{\tau} d/\lambda)$ times (cf. Lemma \ref{lem:connector-bound}). This implies that each node does not have to prepare too many short walks. It turns out that this aspect requires quite a bit more work in the dynamic setting and therefore needs new ideas and techniques. The compact pseudo code is given in Algorithm \ref{alg:single-random-walk}. 

\newcommand{\mindegree}[0]{\delta}
\begin{algorithm}[H]
\caption{\sc Single-Random-Walk($s$, $\tau$)}
\label{alg:single-random-walk}
\textbf{Input:} Starting node $s$, desired walk length $\tau$ and parameter $\lambda$.\\
\textbf{Output:} Destination node of the walk outputs the ID of $s$.\\

\textbf{Phase 1: (Each node $v$ performs $d = \deg(v)$ random walks of length $\lambda + r_i$ where $r_i$ (for each $1\leq i \leq d$) is chosen independently at random in the range $[0, \lambda - 1]$. At the end of the process, there are $d$ (not necessarily distinct) nodes holding a ``coupon" containing the ID of v.)}
\begin{algorithmic}[1]
\FOR{each node $v$}
\STATE  Generate $d$ random integers in the range $[0, \lambda - 1]$, denoted by $r_1, r_2, \ldots,r_{d}$.
\STATE Construct $d$ messages containing its ID, a counter number and in addition, the $i$-th message contains the desired walk length of $\lambda + r_i$. 
We will refer to these messages created by node $v$ as ``coupons created by $v$".
\ENDFOR

\FOR{$i=1$ to $2 \lambda$}

\STATE This is the $i$-th round. Each node $v$ does the following: Consider each coupon $C$ held by $v$ which is received in the $(i - 1)$-th round. If the coupon $C$'s desired walk length is at most $i$, then $v$ keeps this coupon ($v$ is the desired destination). Else, $v$ picks a neighbor $u$ uniformly at random  for each coupon $C$ and forward $C$ to $u$.


\ENDFOR

\end{algorithmic}

\textbf{Phase 2: (Stitch short walks by token forwarding. Stitch $\Theta (\tau/\lambda)$ walks, each of length in $[\lambda, 2 \lambda -1]$.)}
\begin{algorithmic}[1]
\STATE The source node $s$ creates a message called ``token'' which contains the ID of $s$

\STATE The algorithm will forward the token around and keep track of a set of connectors, denoted by $C$. Initially, $C = \{s\}$

\WHILE {Length of walk completed is at most $\tau-2 \lambda$}

  \STATE Let $v$ be the node that is currently holding the token.
  
 \STATE $v$ sample one of the coupons distributed by $v$ uniformly at random (by randomly chosen one counter number from the unused set of coupons). Let $v'$ be the destination node of the sampled coupon, say $C$.






  \STATE $v$ sends the token to $v'$ through broadcast and delete the coupon $C$.  

  \STATE $C = C \cup \{v\}$

\ENDWHILE

\STATE Walk naively until $\tau$ steps are completed (this is at
most another $2 \lambda$ steps)

\STATE A node holding the token outputs the ID of $s$

\end{algorithmic}

\end{algorithm}

\subsection{Analysis}
We first show the correctness of the algorithm and then analyze the time complexity.
\subsubsection{Correctness}
\label{sec:correctness}
\begin{lemma}\label{lem:correctness}
The algorithm {\sc Single-Random-Walk}, with high probability, outputs a node sample that is close to  the uniform probability distribution on the vertex set $V$. 
\end{lemma}
\begin{proof} (sketch)
We know (from  Theorem \ref{thm:mixtime}) that any random walk on a regular evolving graph reaches ``close"  to the uniform distribution at step $\tau$  regardless of any changes of the graph in each round as long as it is $d$-regular, non-bipartite and connected. Therefore it is sufficient to show that {\sc Single-Random-Walk} finishes with a node $v$ which is the destination of a true random walk of length $\tau$ on some appropriate dynamic graph from the source node $s$. We show this below in two steps. \\
First we show that each short walk (of length approximately $\lambda$) created in phase~1 is a true random walk on a dynamic graph sequence $G_1, G_2, \ldots, G_{\tilde{\lambda}}$ ($\tilde{\lambda}$ is some approximate value of $\lambda$). This means that in every step $t$, each walk moves to some random neighbor from the current node on the graph $G_t$ and each walk is independent of others. The proof of the Lemma~\ref{lem:phase1}  shows that w.h.p there is at most $O(\log^2 n)$ bits congestion  in any edge in any round in Phase~1. Since we consider {\em CONGEST}($\log^2 n$) model, at each round $O(\log^2 n)$ bits can be sent through each edge from each direction. Hence effectively there will be no delay in Phase~1 and all walks can extend their length from $i$ to $i+1$ in one round. Clearly each walk is independent of others as every node sends messages independently in parallel. This proves that each short walk (of a random length in the range $[\lambda, 2\lambda-1]$) is a true random walk on the graph $G_1, G_2, \ldots, G_{\tilde{\lambda}}$. \\
In Phase~2, we stitch  short walks to get a long walk of length $\tau$. Therefore, the $\tau$-length random walk is not from the dynamic graph sequence $G_1, G_2, \ldots, G_{\tau}$; rather it is from the sequence:\\ $G_1, G_2, \ldots, G_{\tilde{\lambda}}, G_1, G_2, \ldots, G_{\tilde{\lambda}}, \ldots,$ ($\tau/\lambda$ times approximately). The stitching part is done on the graph sequence from $G_{\tilde{\lambda} +1}, G_{\tilde{\lambda} +2}, \ldots$ onwards. This does not affect the distribution of probability on the vertex set in each step, since the graph sequence from $G_{\tilde{\lambda} +1}, G_{\tilde{\lambda} +2}, \ldots$ is used only for communication. 
Also note that since we define $\tau$ to be the maximum of any static graph $G_t$'s  mixing time, it clearly reaches close to the uniform distribution after $\tau$ steps of walk
 in the graph sequence
 $G_1, G_2, \ldots, G_{\tilde{\lambda}}, G_1, G_2, \ldots, G_{\tilde{\lambda}}, \ldots,$ ($\tau/\lambda$ times approximately).\\
 Finally, when we stitch at a node $v$, we are sampling a coupon (short walk) uniformly at random among many coupons (and therefore, short walks starting at $v$) distributed by $v$. It is easy to see that this stitches  short random walks  independently and hence gives a true random
walk of longer length.
Thus it follows that the algorithm {\sc Single-Random-Walk} returns a destination node of a $\tau$-length random walk (starting from $s$) on some evolving graph.     
\end{proof}

\subsubsection{Time Analysis}
We show the running time of algorithm {\sc Single-Random-Walk}  (cf. Theorem \ref{thm:maintheorem}) using the following lemmas.  
\begin{lemma}\label{lem:phase1}
Phase 1 finishes in $O(\lambda)$ rounds with high probability.  
\end{lemma} 
\begin{proof}
In phase 1, each node $v$ performs $d$ walks of length $\lambda$. Initially all the node starts with $d$ coupons (or messages) and each coupon takes a random walk. We prove that after any given number of steps $j$, the expected number of coupons at node any $v$ is still $d$. Though the edges are changes round to round, but at any round, every node has $d$-neighbors connected with it. So at each step every node can send (as well as receive) $d$ messages. Now the number of messages we started at any node $v$ is proportional to its degree and stationary distribution is uniform here.Therefore, in expectation the number of messages at any node remains same. Thus in expectation the number of messages, say $X$ that go through an edge in any round is at most $2$ (from both end points). Using Chernoff's bound we get ($\Pr[X\geq 4 \log n] \leq 2^{-4\log n} = n^{-4}$). It follows easily from there that the number of messages can go through any edge in any round is at most $4 \log n$ with high probability. Hence there will be at most $O( \log^2 n)$ bits w.h.p. in any edge per round . Since we consider {\em CONGEST}($\log^2 n$) model, so there will be delay due to congestion. Hence, phase 1 finishes in $O(\lambda)$ rounds with high probability.     
\end{proof}
\begin{lemma}\label{lem:lemma2.3}
Sample-Coupon always finishes within $O(\Phi)$ rounds.
\end{lemma} 
\begin{proof}
The proof follows directly from the definition of dynamic diameter $\Phi$. Since one can sample-coupon by at most flooding time and $\Phi$ is maximum of all flooding time of all vertex.   
\end{proof}

We note that the adversary can force the random walk to visit any particular vertex several times. Then we need many short walks from each vertex which increases the round complexity.  We show the following key technical lemma (Lemma~\ref{lem:visit-bound}) that bounds the number of visits to each node in a random walk of length $\ell$.  
In a $d$-regular dynamic graph, we show that no node is visited more than $\tilde{O}(\sqrt{\tau} d/\lambda)$ times as a connector node of a $\tau$-length random walk. For this we need a technical result on  random walks that bounds the number of times a node will be visited in a $\ell$-length (where $\ell = O(\tau)$) random walk. Consider a simple random walk on a connected $d$-regular evolving graphs on n vertices. Let $N_x^t (y)$ denote the number of visits to vertex $y$ by time $t$, given the walk started at vertex $x$. 
Now, consider $k$ walks, each of length $\ell$, starting from (not necessary distinct) nodes $x_1, x_2, \ldots ,x_k$. 

\begin{lemma}\label{lem:visit-bound}
$(${\sc Random Walk Visits Lemma}$)$. For any nodes $x_1, x_2, \ldots, x_k$, \[\Pr\bigl(\exists y\ s.t.\
\sum_{i=1}^k N_\ell^{x_i}(y) \geq 32 \ d \sqrt{k\ell+1}\log n+k\bigr) \leq 1/n\,.\]
\end{lemma}
To prove the above lemma we need to go through some crucial results. We start with the bound of the first moment of the number of visits at each node by each walk.
\begin{proposition}\label{proposition:first-moment} For
any node $x$, node $y$ and $t = O(\tau)$,
\begin{equation}
\e[N_t^x(y)] \le 8 \ d \sqrt{t+1}
\end{equation}
\end{proposition}

To prove the above proposition, let $P$ denote the transition probability matrix of such a random walk and let $\pi$ denote the stationary distribution of the walk. 

The basic bound we use is the estimate from Lyons lemma (see Lemma~3.4 in \cite{Lyons}). We show below that the Lyons lemma also holds for a regular evolving graph. 
\begin{lemma}\label{lem:lyons}
Let $Q$ denote the transition probability matrix of a $d$-regular evolving graph, with self-loop probability $\alpha > 0$. Let $c= \min{\{\pi(x) Q(x,y) : x \neq y \mbox{ and }Q(x,y)>0\}} > 0\,$. Note that here $c = \frac{1}{n d}$, as $\pi$ is uniform distribution. Then for any vertex $x$ and all $k > 0$, a positive integer (denoting time),

$$\bigl|\frac{Q^k(x,x)}{\pi(x)} - 1\bigr| \le
\min\Bigl\{\frac{1}{\alpha c \sqrt{k+1}}, \frac{1}{2\alpha^2 c^2(k+1)} \Bigr\}\,.$$
\end{lemma}


\begin{proof} 
Let $G = (V, E)$ be any $d$-regular graph and $Q$ be the transition probability matrix of it. Write $$c_2(x, y) := \pi(x) Q^2(x, y)$$ and note that for $(x, y) \in E$, we have 
\begin{align*} 
c_2(x, y) \ge \pi(x) \bigl[Q(x, x) Q(x, y) + Q(x, y) Q(y, y)\bigr] \ge 2\alpha c.
\end{align*}
We write $\ell^2(V, \pi)$ for the vector space $\mathbb{R}^{V}$ equipped with the inner product defined by $$(f_1, f_2)_{\pi} := \sum_{x \in V} f_1(x) f_2(x)\pi(x).$$ We regard elements of $\mathbb{R}^{V}$ as functions from $V$ to $\mathbb{R}$. Therefore we will call eigenvectors of the matrix $Q$ as eigenfunctions. Recall that the transition matrix $Q$ is reversible with respect to the stationary distribution $\pi$. The reason for introducing the above inner product is
\begin{claim}\label{lem:lavin}
Let $Q$ be a reversible transition matrix with respect to $\pi$. Then the inner product space $\langle \ell^2(V, \pi), (\cdot , \cdot)_{\pi}\rangle$ has an orthonormal basis of real-valued eigenfunctions $\{f_i \}_{j=1}^{\vert V \vert}$ corresponding to real eigenvalues $\{\lambda_j\}$. 
\end{claim}
\begin{proof}
Denote by $(\cdot , \cdot)$ the usual inner product on $\mathbb{R}^V$, given by $(f_1, f_2) := \sum_{x \in V} f_1(x) f_2(x)$. For a regular graph, $Q$ is symmetric. The more general version proof is given in Lemma~12.2 in~\cite{Levin} where $Q$ need not be symmetric. The spectral theorem for symmetric matrices guarantees that the inner product space $\langle \ell^2(V, \pi), (\cdot , \cdot) \rangle$ has an orthonormal basis $\{\varphi_j\}_{j=1}^{\vert V \vert}$ such that $\varphi_j$ is an eigenfunction with the real eigenvalue $\lambda_j$. It is known that $\sqrt{\pi}$ is an eigenfunction of $Q$ corresponding to the eigenvalue $1$; we set $\varphi_1 = \sqrt{\pi} $ and $\lambda_1 =1$. If $D_{\pi}$ denote the diagonal matrix with diagonal entries $D_{\pi}(x, x) = \pi(x)$, then $Q = D_{\pi}^{\frac{1}{2}} Q D_{\pi}^{-\frac{1}{2}}$. Let $f_j = D_{\pi}^{-\frac{1}{2}} \varphi_j$, then $f_j$ is an eigenfunction of $Q$ with eigenvalue $\lambda_j$. Infact: 
\begin{align*} 
Q f_j = Q D_{\pi}^{-\frac{1}{2}} \varphi_j =  D_{\pi}^{-\frac{1}{2}} (D_{\pi}^{\frac{1}{2}} Q D_{\pi}^{-\frac{1}{2}}) \varphi_j = D_{\pi}^{-\frac{1}{2}} Q \varphi_j  = D_{\pi}^{-\frac{1}{2}} \lambda_j \varphi_j = \lambda_j f_j
\end{align*}
Although the eigenfunctions $\{f_j\}$ are not necessarily orthonormal with respect to the usual inner product, they are orthonormal with respect to the inner product $(\cdot, \cdot)_{\pi}$: 
$$\delta_{ij} = (\varphi_i, \varphi_j) = (D_{\pi}^{\frac{1}{2}} f_i , D_{\pi}^{\frac{1}{2}} f_j) = (f_i, f_j)_{\pi}, $$ the first equality follows since $\{\varphi_j\}$ is orthonormal with respect to the usual inner product. 
\end{proof}

Let $\ell_0^2(V, \pi)$ be the orthogonal complement of the constants in $\ell^2(V, \pi)$. Note that $\textbf{1}$ is an eigenfunction of $Q$ and that $\ell_0^2(V, \pi)$ is invariant under $Q$. Now we show in the following claim that each $f$ has at least one nonnegative value and at least one nonpositive value, such as $f \in \ell^2_0(V, \pi)$. 

\begin{claim}\label{clm:dirichlets}
Let $G$ be an undirected connected $d$-regular graph on $n$ vertices with transition matrix $A_G$. Let $\lambda_1 \geq \lambda_2 \geq \ldots \geq \lambda_n$ be the eigenvalues and $X_1, X_2, \ldots,X_n$ are the corresponding eigenvectors  of $A_G$. Then for each eigenvector $X_i, i= 2,3,  \ldots,n$, $($other than $X_1)$, has at least one negative and at least one positive co-ordinates.  
\end{claim}
\begin{proof}
It is known that $\lambda_1 = 1$ and the normalized eigenvector corresponding to $\lambda_1$ is $X_1 = \frac{1}{\sqrt{n}}(1, 1, \ldots , 1)$. The set of eigenvectors $\{X_i : i = 1, 2, \ldots,n\}$ form a orthonormal basis of the eigenspace. Then for any normalized eigenvector $X \in \{X_i : i=2, 3, \ldots, n \}$, we have $X \perp X^1$. Hence, $\sum_{i=1}^n x_{i}^2 = 1$ and $ \sum_{i=1}^n x_i = 0$, where $X = (x_1, x_2, \ldots ,x_n)$. Let $x_l$ and $x_s$ be the respectively largest and smallest co-ordinates of $X$. Then clearly $ x_l \geq 1/\sqrt{n}$ and $x_s < 0$. 
\end{proof}

Let $x_0$ be a vertex where $\vert f \vert$ achieves its maximum. Then clearly it follows from the Claim~\ref{clm:dirichlets}
\begin{align}
\label{equ-equation1}
\lVert f \rVert_{\infty} & =  \vert f(x_0) \vert \le \frac{1}{2} \sum_{(x,y) \in E} \lvert f(x) - f(y) \rvert \le \frac{1}{2} \sum_{x, y \in V} c_2(x, y) \lvert f(x) - f(y) \rvert/(2 \alpha c)
\end{align} 
where $\lVert \cdot \rVert_{\infty}$ denote the supremum norm and the factor $1/2$ arises from counting each pair $(x,y)$ in each order. Take $f \in \ell^2_0(V, \pi)$. Notice that $\sum_{x,y \in V} c_2(x, y) = \sum_{x \in V} \pi(x) = 1$. Thus, we have from equation~(\ref{equ-equation1}) by using Cauchy-Schwartz inequality that  
\begin{align*}
(2\alpha c)^2 \lVert f \rVert_{\infty}^2 &  \le \frac{1}{2} \sum_{x, y \in V} c_2(x, y) [f(x) - f(y)]^2 \\
& = \frac{1}{2} \sum_{x, y \in V} f(x)^2 \pi(x) Q^2(x ,y) \\ & - \sum_{x, y \in V} f(x) f(y) \pi(x) Q^2(x, y) \\ & + \frac{1}{2} \sum_{x, y \in V} f(y)^2 \pi(x) Q^2(x, y)
\end{align*}
By reversibility, $\pi(x) Q^2(x, y) = \pi(y)Q^2(y, x)$, and the first and last terms above are equal to common value
$$ \frac{1}{2} \sum_{x \in V} f(x)^2 \pi(x) \sum_{y \in V} Q^2(x, y) = \frac{1}{2} \sum_{x \in V} f(x)^2 \pi(x). $$ 
Therefore the above inequality becomes,
\begin{align*}
(2\alpha c)^2 \lVert f \rVert_{\infty}^2 & \le \sum_{x \in V} f(x)^2 \pi(x) - \sum_{x \in V} f(x) \bigl[ \sum_{y \in V} f(y) Q^2(x, y) \bigr] \pi(x) \\
& = (f, f)_{\pi} - (f, Q^2 f)_{\pi} \\
& = ((\mathbb{I} - Q^2)f, f)_{\pi}.
\end{align*}
Alternatively, we may apply~(\ref{equ-equation1}) to the function $\mbox{sgn}(f)f^2$. Using the trivial inequality $$\lvert \mbox{sgn}(s)s^2 - \mbox{sgn}(t)t^2 \rvert \le \lvert s - t \lvert \cdot (\lvert s \rvert + \lvert t \rvert),$$ valid for any real numbers $s$ and $t$, we obtain that  
\begin{align*}
(2\alpha c)^2 \lVert f \rVert_{\infty}^4 & \le \left(\frac{1}{2} \sum_{x, y \in V} c_2(x, y) \lvert f(x) - f(y)\rvert \cdot (\lvert f(x) \rvert + \lvert f(y) \rvert) \right)^2 \\
& \le \left(\frac{1}{2} \sum_{x, y \in V} c_2(x, y) [f(x) - f(y)]^2 \right) \cdot \left(\frac{1}{2} \sum_{x, y \in V} c_2(x, y) [\lvert f(x) \rvert + \lvert f(y) \rvert]^2 \right) \\
& = ((\mathbb{I} - Q^2)f, f)_{\pi} \cdot ((\mathbb{I} + Q^2)\lvert f \rvert, \lvert f \rvert)_{\pi}
\end{align*}
by the Cauchy-Schwartz inequality and same algebra as above. Therefore, if $(f,f)_{\pi} \le 1$, we have $2(\alpha c)^2 \lVert f \rVert_{\infty}^4 \le ((\mathbb{I} - Q^2)f, f)_{\pi}$. \\
Putting both these estimates together, we get 
\begin{equation}
\label{equ-ineque2}
2(\alpha c)^2 \max \{2\lVert f \rVert_{\infty}^2, \lVert f \rVert_{\infty}^4 \} \le ((\mathbb{I} - Q^2)f, f)_{\pi}
\end{equation} 
for $(f,f)_{\pi} \le 1$. Now we show that the above inequality is also holds for the regular evolving graph. 

\begin{claim}\label{clm:claim-norm}
Let $\mathcal{G} = G_1, G_2, \ldots$ be a $d$-regular, connected evolving graph with the same set $V$ of nodes. Let $A_{G_i}$ be the transpose of the transition matrix of $G_i$. Let the column vector $f = (p_1, p_2, \ldots ,p_n)^T$ be any probability distribution on $V$. Then 
$ \lVert (A_{G_{i+1}}A_{G_{i}} \ldots A_{G_1}) f \rVert_\infty \le \lVert (A_{G_i}A_{G_{i-1}}\ldots A_{G_1}) f \rVert_\infty $ for all $i \geq 1$.
\end{claim}
\begin{proof}
It is known that the transition matrix of any regular graph is doubly stochastic and if a matrix $Q$ is doubly stochastic then so is $Q^2$. Let $(A_{G_i}A_{G_{i-1}} \ldots A_{G_1}) f = (p_1^i, p_2^i, \ldots, p_n^i)^T$ and $\lVert (A_{G_i}A_{G_{i-1}} \ldots A_{G_1}) f  \rVert_\infty = \max \{p_l^i : l= 1,2, \ldots, n \} = \vert p_k^i \vert$ (say). Then 
\begin{align*}
(A_{G_{i+1}}A_{G_i} \ldots A_{G_1}) f = \bigl(\sum_{j \in N(1)} a_{1j}p_j^i, \sum_{j\in N(2)} a_{2j}p_j^i, \ldots, \sum_{j\in N(n)} a_{nj}p_j^i  \bigr)^T
\end{align*} where $N(v)$ is the set of neighbors of $v$ and $a_{ij}$ is the $ij$-th entries of the matrix $(A_{G_{i+1}}A_{G_i} \ldots A_{G_1})$. We show that the absolute value of any co-ordinates of  $(A_{G_{i+1}}A_{G_i}\ldots A_{G_1})f $ is $\leq \vert p_k^i \vert$. Infact for any $l$, 
\begin{align*} \vert \sum_{j\in N(l)} a_{lj}p_j^i \vert \leq \sum_{j\in N(l)} \lvert a_{lj}\rvert \lvert p_j^i \rvert & \leq \lvert p_k^i \rvert \sum_{j\in N(l)} a_{lj} = \lvert p_k^i \vert,
\end{align*} since the matrix is doubly stochastic, the last sum is $1$. 
\end{proof}

Now apply the inequality~(\ref{equ-ineque2}) to $Q^l f$ for $l= 0,1, \ldots, k$.  Summing these inequalities and using Claim~\ref{clm:claim-norm} to obtain,
\begin{align*}
(k+1)2(\alpha c)^2\max \{2\lVert Q^k f \rVert_{\infty}^2, \lVert Q^k f \rVert_{\infty}^4 \} & \le 2(\alpha c)^2 \max \{2\sum_{l=0}^k \lVert Q^l f \rVert_{\infty}^2, \sum_{l=0}^k \lVert Q^k f \rVert_{\infty}^4 \} \\
& \le \sum_{l=0}^k ((\mathbb{I} - Q^2)Q^l f, Q^l f)_{\pi} \\ & = \sum_{l=0}^k ((\mathbb{I} - Q^2)Q^{2l} f, f)_{\pi} \\
& = ((\mathbb{I} - Q^{2k + 2})f, f)_{\pi} \le 1
\end{align*}
for $(f,f)_{\pi} \le 1$. This shows that the norm of $Q^k : \ell^2_0(V, \pi) \rightarrow \ell^{\infty}(V)$ is bounded by 
$$\beta_k := \min\{ [(2\alpha c)^2(k+1)]^{-1/2}, [(\alpha c)^2(2k + 2)]^{-1/4} \}. $$
Let $T: \ell^2(V, \pi) \rightarrow \ell^2_0(V, \pi)$ be the orthogonal projection $Tf := f - (f, \textbf{1})_{\pi} \textbf{1}$. Given what we have shown, we see that the norm of $Q^k T : \ell^(V, \pi) \rightarrow \ell^{\infty}(V)$ is bounded by $\beta_k$. By duality, the same bound holds for $TQ^k : \ell^1(V, \pi) \rightarrow \ell^2(V, \pi)$. Therefore by composition of mapping we deduce that the norm of $Q^k T Q^k : \ell^1(V, \pi) \rightarrow \ell^{\infty}(V)$ is at most $\beta_k^2$ and the norm of $Q^k T Q^{k+1} : \ell^1(V, \pi) \rightarrow \ell^{\infty}(V)$ is at most $\beta_k \beta_{k+1}$. Applying these inequalities to $f := \textbf{1}_x/\pi(x)$ gives the required bound.  
\end{proof}

The more general case is proved in Lyons (see
Lemma~3.4 and Remark~4  in \cite{Lyons}). Sometimes, it is more convenient to use the following bound; 
For $k= O(\tau)$ and small $\alpha$, the above can be simplified to the following bound; see Remark~3 in \cite{Lyons}.
\begin{equation}
\label{one_sided_decay} Q^k(x,y)  \le \frac{4\pi(y)}{c \sqrt{k+1}} =
\frac{4d}{\sqrt{k+1}}\,.
\end{equation}

Note that given a simple random walk on a graph $G$, and a
corresponding matrix  $P$, one can always switch to the lazy version
$Q=(I+P)/2$, and interpret it as a walk on graph $G'$, obtained by
adding  self-loops  to vertices in $G$ so as to double the degree of
each vertex. In the following, with abuse of notation we assume our
$P$ is such a lazy version of the original one.

\begin{proof}[Proof of Proposition~\ref{proposition:first-moment}]
Remember that the evolving graph is $\mathcal{G} = G_1, G_2, \ldots$. Let $X_0, X_1, \ldots $ describe the random walk, with $X_i$
denoting the position of the walk at time $i\ge 0$ on $G_{i+1}$, and let
$\bone_A$ denote the indicator (0-1) random variable, which takes
the value 1 when the event $A$ is true. In the following we also use
the subscript $x$ to denote the fact that the probability or
expectation is with respect to starting the walk at vertex $x$.
First the expectation.
\begin{align*}
\e[N_t^x(y)] =  \e_x[  \sum_{i=0}^t \bone_{\{X_i=y\}}] & = \sum_{i=0}^t P^i(x,y) \\
& \le  4 d \sum_{i=0}^t \frac{1}{\sqrt{i+1}} ,  && \mbox{ (using the above inequality  (\ref{one_sided_decay})) } \\
& \le 8 d \sqrt{t+1}\,.
\end{align*}
\end{proof}
Using the above proposition, we bound the number of visits of each walk at each node, as follows. 

\begin{lemma}\label{lemma:whp one walk one node bound}
For $t = O(\tau)$ and any vertex $y \in \mathcal{G}$, the random walk
started at $x$ satisfies:
\begin{equation*}
\Pr\bigl(N^x_t(y) \ge  32  \ d \sqrt{t+1}\log n \bigr) \le \frac{1}{n^2} \,.
\end{equation*}
\end{lemma}
\begin{proof}
First, it follows from the Proposition that
\begin{equation} 
\Pr\bigl(N^x_t(y) \ge  4\cdot 8 \ d \sqrt{t+1}\bigr) \le \frac{1}{4} \,.\label{eq:simple bound}
\end{equation}

For any $r$, let $L^x_r(y)$ be the time that the random walk
(started at $x$) visits $y$ for the $r^{th}$ time. Observe that, for
any $r$, $N^x_t(y)\geq r$ if and only if $L^x_r(y)\leq t$.
Therefore,
\begin{equation}
\Pr(N^x_t(y)\geq r)=\Pr(L^x_r(y)\leq t).\label{eq:visits eq length}
\end{equation}

Let $r^*=32  \ d \sqrt{t+1}$. By \eqref{eq:simple bound} and
\eqref{eq:visits eq length}, $\Pr(L^x_{r^*}(y)\leq t)\leq
\frac{1}{4}\,.$ We claim that
\begin{equation}
\Pr(L^x_{r^*\log n}(y)\leq t)\leq \left(\frac{1}{4}\right)^{\log
n}=\frac{1}{n^2}\,.\label{eq:hp length bound}
\end{equation}
To see this, divide the walk into $\log n$ independent subwalks,
each visiting $y$ exactly $r^*$ times. Since the event $L^x_{r^*\log
n}(y)\leq t$ implies that all subwalks have length at most $t$,
\eqref{eq:hp length bound} follows.
Now, by applying \eqref{eq:visits eq length} again,
\[\Pr(N^x_t(y)\geq r^*\log n) = \Pr(L^x_{r^*\log n}(y)\leq t)\leq
\frac{1}{n^2}\] as desired.
\end{proof}

We now extend the above lemma to bound the number of visits of {\em
all} the walks at each particular node.

\begin{lemma}\label{lemma:k walks one node bound}
For $t = O(\tau)$, and for any vertex $y \in
\mathcal{G}$, the random walk started at $x$ satisfies:
\begin{equation*}
\Pr\bigl(\sum_{i=1}^k N^{x_i}_t(y) \ge  32  \ d \sqrt{kt+1} \log n+k\bigr) \le \frac{1}{n^2} \,.
\end{equation*}
\end{lemma}
\begin{proof}
First, observe that, for any $r$, $$\Pr\bigl(\sum_{i=1}^k
N^{x_i}_t(y) \geq r-k\bigr)\leq \Pr[N^y_{kt}(y)\geq r].$$ To see
this, we construct a walk $W$ of length $kt$ starting at $y$ in the
following way: For each $i$, denote a walk of length $t$ starting at
$x_i$ by $W_i$. Let $\tau_i$ and $\tau'_i$ be the first and last
time (not later than time $t$) that $W_i$ visits $y$. Let $W'_i$ be
the subwalk of $W_i$ from time $\tau_i$ to $\tau_i'$. We construct a
walk $W$ by stitching $W'_1, W'_2, ..., W'_k$ together and complete
the rest of the walk (to reach the length $kt$) by a normal random
walk. It then follows that the number of visits to $y$ by $W_1, W_2,
\ldots, W_k$ (excluding the starting step) is at most the number of
visits to $y$ by $W$. The first quantity is $\sum_{i=1}^k
N^{x_i}_t(y)-k$. (The term `$-k$' comes from the fact that we do not
count the first visit to $y$ by each $W_i$ which is the starting
step of each $W'_i$.) The second quantity is $N^y_{kt}(y)$. The
observation thus follows.

Therefore,
\begin{align*}
& \Pr\bigl(\sum_{i=1}^k N^{x_i}_t(y)\geq 32 \ d
\sqrt{kt+1}\log n + k\bigr) \\ & \leq \Pr\bigl(N^y_{kt}(y)\geq 32 \ d
\sqrt{kt+1}\log n\bigr) \\ & \leq \frac{1}{n^2}
\end{align*}
where the last inequality follows from Lemma~\ref{lemma:whp one walk
one node bound}.
\end{proof}

Now the Random Walk Visits Lemma (cf. Lemma~\ref{lem:visit-bound}) follows immediately from
Lemma~\ref{lemma:k walks one node bound} by union bounding over all
nodes. \\
 
The above lemma says that the number of visits to each node can be bounded.
However, for each node, we are only interested in the case where it is used as a connector (the stitching points). The lemma below shows that the number of visits as a connector can be bounded as well; i.e., if any node appears $t$ times in the walk, then it is likely to appear roughly $t/\lambda$ times as connectors.

\begin{lemma}\label{lem:connector-bound}
For any vertex $v$, if $v$ appears in the walk at most $t$ times then it appears as a connector node at most $t(\log n)^2/\lambda$ times with probability at least $1-1/n^2$.
\end{lemma}
\begin{proof}
Intuitively, this argument is simple, since the connectors are spread out in steps of length approximately $\lambda$. However, there might be some periodicity that results in the same node being visited multiple times but exactly at $\lambda$-intervals. To overcome this we crucially use the fact that the algorithm uses short walks of length $\lambda + r$ (instead of fixed length $\lambda$) where $r$ is chosen uniformly at random from $[0, \lambda -1]$. Then the proof can be shown via constructing another process equivalent to partitioning the $\tau$ steps into intervals of $\lambda$ and then sampling points from each interval. The detailed proof follows immediately from the proof of the Lemma~2.7 in~\cite{DasSarmaNPT10}.
\end{proof}

Now we are ready to proof the main result (Theorem~\ref{thm:maintheorem}) of this section. \\ 

\noindent\textbf{Proof of the Theorem \ref{thm:maintheorem} (restated below)}
\begin{theorem}
The algorithm {\sc Single-Random-walk} (cf. Algorithm \ref{alg:single-random-walk}) solves the Single Random Walk problem and with high probability finishes in $\tilde{O}(\sqrt{\tau \Phi})$ rounds. 
\end{theorem}
\begin{proof}
First, we claim, using Lemma \ref{lem:visit-bound} and
\ref{lem:connector-bound}, that each node is used as a connector node
at most $\frac{32 \ d \sqrt{\tau}(\log n)^3}{\lambda}$ times with
probability at least $1-2/n$. To see this, observe that the claim
holds if each node $x$ is visited at most
$t(x)=32 \ d \sqrt{\tau+1}\log n$ times and consequently appears as a
connector node at most $t(x)(\log n)^2/\lambda$ times. By
Lemma~\ref{lem:visit-bound}, the first condition holds with
probability at least $1-1/n$. By Lemma~\ref{lem:connector-bound} and
the union bound over all nodes, the second condition holds with
probability at least $1-1/n$, provided that the first condition
holds. Therefore, both conditions hold together with probability at
least $1-2/n$ as claimed.

Now, we choose $\lambda=32 \sqrt{\tau \Phi}(\log n)^3$.
By Lemma~\ref{lem:phase1}, Phase~1 finishes in $O(\lambda) = \tilde O(\sqrt{\tau \Phi})$ rounds with high probability.
For Phase~2, {\sc Sample-Coupon} is invoked
$O(\frac{\tau}{\lambda})$ times (only when we stitch the walks) and
therefore, by Lemma~\ref{lem:lemma2.3}, contributes
$O(\frac{\tau \Phi}{\lambda})=\tilde O(\sqrt{\tau \Phi})$ rounds.

Therefore, with probability at least $1-2/n$, the rounds are $\tilde
O(\sqrt{\tau \Phi})$ as claimed.
\end{proof}

\subsection{Generalization to non-regular evolving graphs}
\label{sec:nonregular}
By using a {\em lazy} random walk strategy, we can generalize our results to work for a non-regular dynamic graph also. The lazy random walk strategy ``converts"  a random walk on an non-regular graph to a slower random walk on a regular graph. 

\begin{definition}\label{def:lazy-rw}
At each step of the walk pick a vertex $v$ from $V$ uniformly at random and if there is an edge from the current vertex to the vertex $v$ then we move to $v$, otherwise we stay at the current vertex. 
\end{definition} 

This strategy of lazy random walk in fact makes the graphs $n$-regular: every edge adjacent to the current vertex is picked with the probability $1/n$ and with the remaining probability we stay at the current vertex. 
Using this strategy, we can obtain the same results on non-regular graphs as well, but with a factor of $n$
slower. In fact, we can do better, if nodes know an an upper bound $d_{max}$ on the maximum degree of the dynamic network. Modify the lazy walk such that at each step of the walk stay at the current vertex $u$ with probability $1 - (d(u)/(d_{max} + 1))$ and with the remaining probability pick a neighbors uniformly at random. This only results in a slow down by a factor of $d_{max}$ compared to the regular case. 
 
\section{Algorithm for $k$ Random Walks}\label{sec:k-algo}
The previous section was devoted to performing a single random walk of length $\tau$ (mixing time) efficiently to sample from the stationary distribution. In many applications, one typically requires a large number of random walk samples. A larger amount of samples allows for a better estimation of the problem at hand. In this section we focus on obtaining several random walk samples.  Specifically, we consider the scenario when we want to compute $k$ independent walks each of
length $\tau$ from different (not necessarily distinct) sources $s_1, s_2, \ldots, s_k$. We show that {\sc Single-Random-Walk} (cf. Algorithm~\ref{alg:single-random-walk}) can be extended to solve this problem. In particular, the algorithm {\sc Many-Random-Walks} (for pseudocode cf. Algorithm~\ref{alg:many-random-walk}) to compute $k$ walks is essentially repeating the {\sc Single-Random-Walk} algorithm on each source with one common/shared phase, and yet through overlapping computation, completes faster than $k$ times the previous bound. The crucial observation is that we have to do Phase 1 only once and still ensure all walks are independent. The high level analysis is following. 

\paragraph{{\sc Many-Random-Walks} :} Let $\lambda=(32 \sqrt{k\tau \Phi+1}\log n+k)(\log n)^2$. If
$\lambda \ge \tau$ then run the naive random walk algorithm. 
Otherwise, do the following. First, modify Phase~2 of {\sc Single-Random-Walk} to create multiple walks, one at a time; i.e., in the second phase, we stitch the short walks together to get a
walk of length $\tau$ starting at $s_1$ then do the same thing for $s_2$, $s_3$, and so on. We show that {\sc Many-Random-Walks} algorithm finishes in $\tilde O\left(\min(\sqrt{k\tau \Phi}, k+\tau)\right)$ rounds with high probability. This result is also stated in the Theorem \ref{thm:kwalks} (Section \ref{sec:results}), but the formal proof is given below. The details of this specific extension is similar to the previous ideas even for the dynamic setting.

\subsection{Proof of the Theorem \ref{thm:kwalks} (restated below)}
\begin{theorem} {\sc Many-Random-Walks} (cf. Algorithm~\ref{alg:many-random-walk}) finishes in
$\tilde O\left(\min(\sqrt{k\tau \Phi}, k+\tau)\right)$
rounds with high probability.
\end{theorem}
\begin{proof}
Recall that we assume $\lambda=(32 \sqrt{k\tau \Phi+1}\log n+k)(\log n)^2$. First, consider the case where $\lambda \ge \tau$. In this case, $\min(\sqrt{k\tau \Phi}+k, \sqrt{k\tau}+k+\tau)=\tilde O(\sqrt{k\tau}+k+\tau)$. By Lemma~\ref{lem:visit-bound}, each
node $x$ will be visited at most $\tilde O(d (\sqrt{k\tau}+k))$ times. Therefore, using the same argument as Lemma~\ref{lem:phase1},
the congestion is $\tilde O(\sqrt{k\tau} + k)$ with high probability. Since the dilation is $\tau$, {\sc Many-Random-Walks}
takes $\tilde O(\sqrt{k\tau} + k +\tau)$ rounds as claimed. Since $2\sqrt{k\tau} \le k + \tau$, this bound reduces
to $\tilde O(k +\tau)$. 

Now, consider the other case where $\lambda < \tau$. In this case,
$\min(\sqrt{k\tau \Phi} +k, \sqrt{k\tau}+k+\tau)=\tilde O(\sqrt{k\tau \Phi}+k)$. Phase~1 takes $O(\lambda) = \tilde O(\sqrt{k\tau \Phi}+k)$. The stitching in Phase~2 takes $\tilde O(k \Phi\tau /\lambda) = \tilde O(\sqrt{k\tau \Phi})$. Since $k \Phi\tau /\lambda \geq k\Phi \geq k$, the
total number of rounds required is $\tilde O(\sqrt{k\tau \Phi})$ as claimed.
\end{proof}

\begin{algorithm}[H]\label{many-walks algorithm}
\caption{\sc Many-Random-Walks($\{s_j \}$, $1\leq j \leq k$, $\tau$)}
\label{alg:many-random-walk}
\textbf{Input:} Starting nodes $s_1, s_2, \ldots, s_k$, (not necessarily distinct) and desired walks length $\tau$ and parameter $\lambda$.\\
\textbf{Output:} Each destination node of the walks outputs the ID of its corresponding source.\\

\textbf{Case~1.} When $\lambda \ge \tau$. [we assumed $\lambda=(32 \sqrt{k\tau \Phi+1}\log n+k)(\log n)^2$]
\begin{algorithmic}[1] 
\STATE  Run the naive random walk algorithm, i.e., the sources find walks of length $\tau$ simultaneously by sending tokens.

\end{algorithmic}

\textbf{Case~2.} When $\lambda < \tau$. \\
\textbf{Phase 1: (Each node $v$ performs $d$ random walks of length $\lambda + r_i$ where $r_i$ (for each $1\leq i \leq d$) is chosen independently at random in the range $[0, \lambda -1]$. At the end of the process, there are $d$ (not necessarily distinct) nodes holding a ``coupon" containing the ID of $v$.)}
\begin{algorithmic}[1]
\FOR{each node $v$}
\STATE  Perform $d$ walks of length $\lambda + r_i$, as in Phase~1 of algorithm {\sc Single-Random-Walk}. 
\ENDFOR

\end{algorithmic}

\textbf{Phase 2: (Stitch $\Theta (\tau/\lambda)$ short walks for each source node $s_j$)}
\begin{algorithmic}[1]
\FOR{j = 1 to k}
\STATE  Consider source $s_j$. Use algorithm {\sc Single-Random-Walk} to perform a walk of length $\tau$ from $s_j$.
\STATE When algorithm {\sc Single-Random-walk} terminates, the sampled destination outputs ID of the source $s_j$. 
\ENDFOR
\end{algorithmic}

\end{algorithm}
\section{Applications}\label{sec:apps}
While the previous sections focused on performing the fundamental primitive of random walks efficiently in a dynamic network, in this section we show that these techniques actually directly help in specific applications in dynamic networks as well. 
\subsection{Information Dissemination (or $k$-Gossip)}
\label{sec:info-dissem}
We present a fully distributed algorithm for the {\em $k$-gossip} problem in $d$-regular evolving graphs (full pseudocode is given in Algorithm \ref{alg:token-dissemination}). 
Our distributed algorithm is based on the centralized algorithm of \cite{DPRS-arxiv} which consists of two phases. The first phase consists
of sending some $f$ copies (the value of the parameter $f$ will be fixed in the analysis) of each of the $k$ tokens to a set of {\em random} nodes. 
 We use algorithm {\sc Many-Random-Walk} (cf. algorithm~\ref{alg:many-random-walk}) to efficiently do this. In the second phase we simply broadcast each token $t$ from the random places to reach all the nodes. We show that if every node having a token $t$ broadcasts it for $O(n\log n/f)$ rounds, then with high probability all the nodes will receive the token $t$.
 
\begin{algorithm}[H]
\caption{\sc K-Information-Dissemination($\mathcal{G}$, $k$)}
\label{alg:token-dissemination}
\textbf{Input:} An evolving graphs $\mathcal{G}: G_1, G_2, \ldots$ and $k$ token in some nodes.\\
\textbf{Output:} To disseminate $k$ tokens to all the nodes.\\

\textbf{Phase 1: (Send $f = n^{2/3} (k/\tau \Phi)^{1/3}$ copies of each token to random places)}
\begin{algorithmic}[1]
\STATE  Every node holding token $t$, send $f = n^{2/3} (k/\tau \Phi)^{1/3}$ copies of each token to random nodes using algorithm {\sc Many-Random-Walk}.

\end{algorithmic}

\textbf{Phase 2: (Broadcast each token for $O(n\log n/f)$ rounds)}
\begin{algorithmic}[1]
\FOR{each token $t$}
\STATE  For the next $2 n\log n/f$ rounds, let all the nodes has token $t$ broadcast the token.
\ENDFOR
\end{algorithmic}

\end{algorithm}

We show that our proposed $k$-gossip algorithm finishes in $\tilde{O}(n^{1/3}k^{2/3}(\tau \Phi)^{1/3})$ rounds w.h.p. 
To make sure that the algorithm terminates in $O(nk)$ rounds, 
we run the above algorithm in parallel with the trivial algorithm (which is just broadcast each of the $k$ tokens sequentially; clearly this will take $O(nk)$ rounds in total) and stops when one of the two algorithm stop. Thus the claimed bound in Theorem \ref{thm:token-bound} holds. The formal proof is below.\\

\noindent \textbf{Proof of the Theorem \ref{thm:token-bound} (restated below)}
\begin{theorem}
The algorithm~(cf. algorithm~\ref{alg:token-dissemination}) solves $k$-gossip problem with high probability \\ in $\tilde{O}(\min\{n^{1/3}k^{2/3}(\tau \Phi)^{1/3}, nk\})$ rounds. 
\end{theorem}
\begin{proof}
We are running both the trivial and our proposed algorithm in parallel. Since the trivial algorithm finishes in $O(nk)$ rounds, therefore we concentrate here only on the round complexity of our proposed algorithm. \\
We are sending $f$ copies of each $k$ token to random nodes which means we are sampling $k f$ random nodes from uniform distribution. So using the {\sc Many-Random-Walk} algorithm, phase 1 takes $\tilde{O}(\sqrt{k f \tau \Phi})$ rounds. 

Now fix a node $v$ and a token $t$. Let $S$ be the set of nodes which has the token $t$ after phase 1. Since the token $t$ is broadcast for $2 n \log n/f$ rounds, there is a set $S_v^t$ of atleast $2 n \log n/f$ nodes from which $v$ is reachable within $2 n \log n/f$ rounds. This is follows from the fact that at any round at least one uninformed node will be informed as the graph being always connected. It is now clear that if $S$ intersects $S_v^t$, $v$ will receive token $t$. The elements of the set $S$ were sampled from the vertex set through the algorithm {\sc Many-Random-Walk} which sample nodes from close to uniform distribution, not from actual uniform distribution. We can make it though very close to uniform by extending the walk length multiplied by some constant. Suppose {\sc Many-Random-Walk} algorithm samples nodes with probability $1/n \pm 1/n^2$ which means each node in $S$ is sampled with probability $1/n \pm 1/n^2$. So the probability of a single node $w \in S$ does not intersect $S^t_v$ is at most $(1 - |S^t_v|(\frac{1}{n} \pm \frac{1}{n^2})) = (1 - \frac{2n\log n}{f} \times \frac{n \pm 1}{n^2})$. Therefore the probability of any of the $f$ sampled node in $S$ does not intersect $S^t_v$ is at most $(1 - \frac{2(n \pm 1)\log n}{n f})^f \leq \frac{1}{n^{2 \pm 2/n}}$. Now using union bound we can say that every node in the network receives the token $t$ with high probability. This shows that phase~2 uses $k n\log n/f$ rounds and sends all $k$ tokens to all the nodes with high probability. Therefore the algorithm finishes in $\tilde{O}(\sqrt{k f \tau \Phi} + k n/f)$ rounds. Now choosing $f = n^{2/3} (k/\tau \Phi)^{1/3}$ gives the bound as $\tilde{O}(n^{1/3} k^{2/3} (\tau \Phi)^{1/3})$. Hence, the $k$-gossip problem solves with high probability in $\tilde{O}(\min\{n^{1/3}k^{2/3}(\tau \Phi)^{1/3}, nk\})$ rounds. 
\end{proof}
Note that the mixing time $\tau$ of a regular dynamic graph is at most $O(n^2)$ (follows from Theorem~\ref{thm:mixtime} and Corollary~\ref{cor:second-eigen-bound}). Putting this in Theorem  \ref{thm:token-bound}, yields a better  bound for $k$-gossip problem in a regular dynamic graph.


\subsection{Decentralized Estimation of Mixing Time}
\label{sec:mixest}
We focus on estimating the {\em dynamic mixing time} $\tau$ of a $d$-regular connected non-bipartite evolving graph $\mathcal{G}= G_1, G_2, \ldots$. We discussed in Section \ref{mixing_time} that $\tau$ is maximum of the mixing time of any graph in $\{G_t : t \geq 1 \}$. To make it appropriate for our algorithm, we will assume that all graphs $G_t$ in the graph process $\mathcal{G}$ have the same mixing time $\tau_{mix}$. Therefore $\tau = \tau_{mix}$. While the definition of $\tau$ (cf. Definition \ref{def:mix-dynamic}) itself is consistent, estimating this value becomes significantly harder in the dynamic context. The intuitive approach of estimating distributions continuously and then adapting a distribute-closeness test works well for static graphs, but each of these steps becomes far more involved and expensive when the network itself changes and evolves continuously. Therefore we need careful analysis and new ideas in obtaining the following results. We introduce related notations and definitions in Section~\ref{mixing_time}. 

The goal is to estimate $\tau^x_{mix}$ (mixing time for source $x$). Notice that the definition $\tau^x_{mix}$ and dynamic mixing time, $\tau$ (cf. Section \ref{mixing_time}) are consistent for a $d$-regular evolving graph $\mathcal{G} = G_1,G_2, \dots$ due to the monotonicity property (cf. Lemma~\ref{lem:monotonicity}) of distributions.

We now present an algorithm to estimate $\tau$. The main idea behind this approach is, given a source node, to run many random walks of some length $\ell$ using the approach described in Section \ref{sec:k-algo}, and use these to estimate the distribution induced by the $\ell$-length random walk. We then compare the the distribution at length $\ell$, with the stationary distribution to determine if they are close, and if not, double $\ell$ and retry.     

For the case of static graph (with diameter $D$), Das Sarma et al.~\cite{DasSarmaNPT10} shows that the one can approximate mixing time in $\tilde O(n^{1/4} \sqrt{D\tau^x(\epsilon)})$ rounds. We show here that this bound also holds to approximate mixing time even for the dynamic graphs which is $d$-regular.  We use the technique of Batu et al.~\cite{BFFKRW} to determine if the distribution is $\epsilon$-near to uniform distribution. Their result is restated in the following theorem. 

\begin{theorem}[\cite{BFFKRW}]\label{thm:batu}
For any $\epsilon$, given $\tilde{O}(n^{1/2}poly(\epsilon^{-1}))$ samples of a distribution $X$
over $[n]$, and a specified distribution $Y$, there is a test that outputs PASS with high probability if $|X-Y|_1\leq \frac{\epsilon^3}{4\sqrt{n}\log n}$, and outputs FAIL with high probability if $|X-Y|_1\geq 6\epsilon$.
\end{theorem}

The distribution $X$ in our context is some distribution on nodes and $Y$ is the stationary distribution, i.e., $Y(v) = 1/n$ (assume $\vert V \vert = n$ in the
network). 
We now give a very brief description of the algorithm of Batu et. al.~\cite{BFFKRW} to illustrate that it can in fact be simulated on the distributed network efficiently. The algorithm partitions the set of nodes in to buckets based on the steady state probabilities. Each of the $\tilde{O}(n^{1/2}poly(\epsilon^{-1}))$ samples from $X$ now falls in one of these buckets. Further, the actual count of number of nodes in these buckets for distribution $Y$ are counted. The exact count for $Y$ for at most $\tilde{O}(n^{1/2}poly(\epsilon^{-1}))$ buckets (corresponding to the samples) is compared with the number of samples from $X$; these are compared to determine if $X$ and $Y$ are close. Note that the total number of nodes n and $\epsilon$ can be broadcasted to all nodes in $O(\Phi)$ rounds and each node can determine which bucket it is in in $O(\Phi)$ rounds.We refer the reader to their paper~\cite{BFFKRW} for a precise description.

Our algorithm starts with $\ell=1$ and runs $K=\tilde{O}(\sqrt{n})$ walks of length $\ell$ from the specified source $x$. As the test of comparison with the steady state distribution outputs FAIL (for choice of $\epsilon=1/12e$), $\ell$ is doubled. This process is repeated to identify the largest $\ell$ such that the test outputs FAIL with high probability and the smallest $\ell$ such that the test outputs PASS with high probability. These give lower and upper bounds on the required $\tau^x_{mix}$ respectively. Our resulting theorem is presented below. \\

\noindent \textbf{Proof of the Theorem \ref{thm:complexity_bound_mixing_time} (restated below)}
\begin{theorem}
Given connected $d$-regular evolving graphs with dynamic diameter $\Phi$, a node $x$ can find, in $\tilde{O}(n^{1/4}\sqrt{\Phi \tau^x(\epsilon)})$ rounds, a time
$\tilde{\tau}^x_{mix}$ such that $\tau^x_{mix}\leq \tilde{\tau}^x_{mix} \leq \tau^x(\epsilon)$, where $\epsilon = \frac{1}{6912e\sqrt{n}\log n}$.
%
%
\end{theorem}
\begin{proof}
Our goal is to check when the probability distribution (on vertex set $V$) of the random walk becomes stationary distribution which is uniform here. If a source node knows the total number of nodes in the network (which can be done through flooding in $O(\Phi)$ rounds), we only need
$\tilde{O}(n^{1/2}poly(\epsilon^{-1}))$ samples from a distribution to
compare it to the stationary distribution.  This can be achieved by
running {\sc MultipleRandomWalk} to obtain $K = \tilde{O}(n^{1/2}poly(\epsilon^{-1}))$ random walks. We choose $\epsilon = 1/12e$.
To find the approximate mixing time, we try out
increasing values of $\ell$ that are powers of $2$.  Once we find the
right consecutive powers of $2$, the monotonicity property admits a
binary search to determine the exact value for the specified $\epsilon$.

We have shown previously that a source node can obtain $K$ samples from $K$ independent random walks of length $\ell$ in $\tilde{O}(\sqrt{K\ell \Phi})$ rounds. Setting $K=\tilde{O}(n^{1/2}poly(\epsilon^{-1}))$ completes the proof.
\end{proof}

Suppose our estimate of $\tau^x_{mix}$ is close to the dynamic mixing time of the network defined as $\tau = \max_{x}{\tau^x_{mix}}$, then this would allow us to estimate several related quantities. Given a dynamic mixing time $\tau$, we can approximate the spectral gap ($1-\lambda$) and the conductance ($\Psi$) due to the
known relations that $\frac{1}{1-\lambda}\leq \tau \leq \frac{\log n}{1-\lambda}$ and $\Theta(1-\lambda)\leq \Psi \leq \Theta(\sqrt{1-\lambda})$ as shown in~\cite{JS89}. 

\section{Conclusion}\label{sec:conc}
We presented fast and fully decentralized algorithms for performing several random walks in distributed dynamic networks. Our algorithms satisfy strong round complexity guarantees and is the first work to present robust techniques for this fundamental graph primitive in dynamic graphs.
 We further extend the work to show how it can be used for efficient sampling and other applications such as token dissemination. 
Our work opens several interesting research directions. 
 In the recent years, several fundamental graph operatives are being explored in various distributed dynamic models, and it would be interesting to explore further along these lines and obtain new approaches for identifying sparse cuts or graph partitioning, and similar spectral quantities. 
As a specific question, it remains open whether the random walk techniques and subsequent bounds presented in this paper are optimal. 
 Finally, these algorithmic ideas may be useful building blocks in designing fully dynamic self-aware distributed graph systems. 
 It would be interesting to additionally consider total message complexity costs for these algorithms explicitly, even though they are implicitly encapsulated within the local per-edge bandwidth constraints of the CONGEST model. 

\newpage

  \let\oldthebibliography=\thebibliography
  \let\endoldthebibliography=\endthebibliography
  \renewenvironment{thebibliography}[1]{%
    \begin{oldthebibliography}{#1}%
      \setlength{\parskip}{0ex}%
      \setlength{\itemsep}{0ex}%
  }%
  {%
    \end{oldthebibliography}%
  }
{
\bibliographystyle{abbrv}
\bibliography{Distributed-RW}
}

\end{document}